\documentclass[a4paper,UKenglish,cleveref,autoref,final]{lipics-v2019} 

\usepackage[utf8]{inputenc}
\usepackage{paralist}
\usepackage{amsmath,amssymb}
\usepackage[obeyFinal]{todonotes}

\usepackage{tikz}
\usepackage{pgflibrarysnakes}
\usetikzlibrary{fit,calc,positioning,decorations,automata}
\tikzset{>=latex,every state/.style={inner sep=1pt,minimum size=2em}}

\newcommand{\dyn}[1]{\ensuremath{\xrightarrow{\mathsf{#1}}}}

\newcommand{\Odyn}{\dyn{1}}
\newcommand{\POdyn}{\dyn{P1}}

\newcommand{\PCdyn}{\dyn{PC}}
\newcommand{\BPCdyn}{\dyn{bPC}}
\newcommand{\BPOdyn}{\dyn{bP1}}

\newcommand{\graph}[2]{#1\langle #2\rangle}

\newcommand{\eq}{\sim}

\usepackage{xspace}
\newcommand\OTG{\ensuremath{\mathrm{1TG}}\xspace}
\newcommand\NOTG{\ensuremath{\mathrm{N1TG}}\xspace}

\newcommand\dom{\mathrm{dom}}

\newcommand{\game}{\mathcal{G}}

\newcommand{\play}{\mathrm{Play}}
\newcommand{\hist}{\mathrm{Hist}}
\newcommand{\last}[1]{\mathop{\mathrm{last}}(#1)}
\newcommand{\first}[1]{\mathop{\mathrm{first}}(#1)}
\newcommand{\player}[1]{\mathop{\mathrm{player}}(#1)}
\newcommand{\playP}{\mathrm{Play}^{\mathsf{P}}}

\newcommand{\Outc}[2]{\ensuremath{\mathsf{Outcome}\left(#1,#2\right)}}

\newcommand{\pstrat}{\ensuremath{\Sigma}^\mathsf{P}}

\newcommand{\DWAB}{Strong Dispute Wheel}

\newcommand{\Target}{\ensuremath{V_{\bot}}}
\newcommand{\target}{\ensuremath{v_{\bot}}}
\newcommand{\perm}{\ensuremath{\mathcal{P}}}
\newcommand{\nperm}{\ensuremath{\mathcal{P}^c}}

\newcommand{\Lfair}{$L$-fair}

\newcommand\edge[2]{#2_{#1}}

\usepackage{xparse}

\title{Dynamics on Games: Simulation-Based Techniques and Applications
to Routing} 

\author{Thomas Brihaye}{Université de Mons, Mons,
  Belgium}{thomas.brihaye@umons.ac.be}{
}{
}

\author{Gilles Geeraerts}{Université libre Bruxelles, Brussels, Belgium}{gigeerae@ulb.ac.be}{
}{
}

\author{Marion Hallet}{Université de Mons, Mons,
  Belgium\\Université libre Bruxelles, Brussels, Belgium}{marion.hallet@umons.ac.be}{
}{
}

\author{Benjamin Monmege}{Aix Marseille Univ, CNRS, LIS, Université de Toulon,
    France}{benjamin.monmege@univ-amu.fr}{https://orcid.org/0000-0002-4717-9955}{}
  
\author{Bruno Quoitin}{Université de Mons, Mons,
  Belgium}{bruno.quoitin@umons.ac.be}{}{}

\authorrunning{T. Brihaye, G. Geeraerts, M. Hallet, B. Monmege, B. Quoitin}

\Copyright{Thomas Brihaye, Gilles Geeraerts, Marion Hallet, Benjamin Monmege, Bruno Quoitin}


\funding{
  G.~Geeraerts is supported by an ARC grant of the Fédération
  Wallonie-Bruxelles; M.~Hallet is supported by an FNRS grant and an
  UMONS grant of the Fonds Franeau Mobilité; B.~Monmege is partially
  supported by the DeLTA project (ANR-16-CE40-0007) and by ANR project
  Ticktac (ANR-18-CE40-0015).  This article is based upon work from
  COST Action GAMENET CA 16228 supported by COST (European Cooperation
  in Science and Technology). }

\acknowledgements{We thank Timothy Griffin and Marco Chiesa for fruitful discussions.}

\EventEditors{Arkadev Chattopadhyay and Paul Gastin}
\EventNoEds{2}
\EventLongTitle{39th Annual Conference on
Foundations of Software Technology and Theoretical Computer Science
(FSTTCS'19)}
\EventShortTitle{FSTTCS 2019}
\EventAcronym{FSTTCS}
\EventYear{2019}
\EventDate{December 11--13, 2019}
\EventLocation{Bombay, India}
\EventLogo{}
\SeriesVolume{42}
\ArticleNo{23}

\EventEditors{Arkadev Chattopadhyay and Paul Gastin}
\EventNoEds{2}
\EventLongTitle{39th IARCS Annual Conference on Foundations of Software Technology and Theoretical Computer Science (FSTTCS 2019)}
\EventShortTitle{FSTTCS 2019}
\EventAcronym{FSTTCS}
\EventYear{2019}
\EventDate{December 11--13, 2019}
\EventLocation{Bombay, India}
\EventLogo{}
\SeriesVolume{150}
\ArticleNo{35}

\begin{document}

\maketitle

\begin{abstract}
  We consider multi-player games played on graphs, in which the
  players aim at fulfilling their own (not necessarily antagonistic)
  objectives. In the spirit of evolutionary game theory, we suppose
  that the players have the right to repeatedly update their
  respective strategies (for instance, to improve the outcome w.r.t.\
  the current strategy profile). This generates a dynamics in the game
  which may eventually stabilise to an equilibrium. The objective of
  the present paper is twofold. First, we aim at drawing a general
  framework to reason about the termination of such
  dynamics. In~particular, we identify preorders on games (inspired
  from the classical notion of simulation between transitions systems,
  and from the notion of graph minor) which preserve termination of
  dynamics. Second, we show the applicability of the previously
  developed framework to interdomain routing~problems.
\end{abstract}

\section{Introduction}

Games are nowadays a well-established model to reason about several
problems in computer science. In the game paradigm, several agents
(called \emph{players}) are assumed to be rational, and interact in
order to reach a fixed objective. As such, games have found numerous
applications, such as controller synthesis \cite{PR89,Th95} or network
protocols \cite{NRTV07}.  In this paper, we are mainly concerned about
\emph{multi-player games played on graphs}, in which $n\geq 2$ players
interact trying to fulfil their own objectives (which are not
necessarily antagonistic to the others); and where the arena
(defining the possible actions of the players) is given as a finite
graph.

An example of such game is given in
\figurename~\ref{Fig:DISAGREE+graphDyn}, modelling an instance of an
\emph{interdomain routing problem} which is typical of the Internet.
In this case, two service providers $v_1$ and $v_2$ want to route
packets to a target node $v_\bot$ through the links that are
represented by the graph edges. For economical reasons, $v_1$ prefers
to route the traffic to $v_\bot$ through $v_2$ (using path $c_1s_2$)
instead of sending them directly to $v_\bot$, and symmetrically for
$v_2$. (Assume for instance that both $v_1$ and $v_2$ are located in
Europe, and that $v_\bot$ is in America. Then, $s_1$ and $s_2$ are
transatlantic links that incur a huge cost of operation for the origin
nodes.) Then, assume that, initially, $v_1$ and $v_2$ route the
packets through $s_1$ and $s_2$ respectively, and broadcast this
information through the network. When $v_1$ becomes aware of the
choice of $v_2$, he could decide to rely on the $c_1$ link instead,
trying to route his packets through $v_2$. However, due to the
asynchronous nature of the network, $v_2$ could decide to route
through $c_2$ before the new choice of $v_1$ reaches it. Hence, the
packets get blocked in a cycle $c_1c_2c_1\cdots$ and do not reach
$v_\bot$ anymore. Then, $v_1$ and $v_2$ could decide simultaneously to
reverse to $s_1$ and $s_2$ respectively which brings the network to
its initial state, where the same behaviour can start again.
Clearly, such oscillations in the routing policies must be avoided.

This simple example illustrates the main notions we will consider in
the paper. We study the notion of \emph{dynamics in games}, which
model the behaviour of the players when they repeatedly update their
strategy (i.e.\ their choices of actions) in order to achieve a better
outcome.  Then, the main objectives of the paper are to \emph{draw a
  general framework to reason about the termination of such dynamics}
and to \emph{show its applicability to interdomain routing problems}
(as sketched above). We say that a dynamics terminates when the
players converge to an \emph{equilibrium}, i.e.\ a state in which they
have no incentive to further update their respective strategies. Our
framework is introduced in Section~\ref{Sec:Graphs}
and~\ref{Sec:Game}. It relies on notions of \emph{preorders}, in
particular the \emph{simulation} preorder~\cite{Milner}. Simulations
are usually defined on transition systems: intuitively, a system $A$
simulates a system $B$ if each step of $B$ can be mimicked in $A$. We
consider two kinds of preorders: preorders defined on \emph{game
  graphs}, i.e.\ on the structure of the games; and simulation defined
on the \emph{dynamics}, which are useful to reason about termination
(indeed, if a dynamics $D_1$ simulates a dynamics $D_2$, and if $D_1$
terminates, then $D_2$ terminates as well). We show how the existence
of a relation between \emph{game graphs} implies the existence of a
simulation between the \emph{induced dynamics} of those games
(Theorem~\ref{Thm:minor of game=>simulation of graph},
Theorem~\ref{theo:fair-termination}). This technique allows us to
\emph{check the termination of the dynamics using structural criteria
  about the game graph}.

The motivation of this framework comes from several examples of
problems in the literature
\cite{Griffin02,SamiSchapiraZohar,LeRouxDynamics,articleGandalf} that
are (sometimes implicitly) reduced to checking the termination of a
dynamics in a multi-player game, and where sufficient criteria are
proposed that can be expressed as the existence of a preorder between
game graphs.  We thus seek to unify these results, hoping that our
framework will foster new applications of the game model. For
instance, several sufficient conditions for termination in the network
problem sketched above consist in checking that the game graph does
not contain a \emph{forbidden pattern}~\cite{Griffin02}. This
containment can naturally be expressed as a preorder.

To this aim, we introduce, in Section~\ref{Sec:Game} a preorder
relation on game graphs, which is inspired from the classical notion
of \emph{graph minor} \cite{Lo06}. Intuitively, a game graph $\game'$
is a minor of $\game$ if $\game'$ can be obtained by deleting edges
and vertices from $\game$ (under well-chosen conditions that are
compatible with the game setting). Then, the relation `is a minor of'
forms a preorder relation on game graphs and allows one to reason on
the termination of dynamics (see Theorem~\ref{Thm:minor of
  game=>simulation of graph} and Theorem~\ref{theo:fair-termination}).

Finally, in Section~\ref{sec:applications}, we achieve our second
objective, by casting questions about Interdomain Routing into our
framework. Interdomain Routing is the process of constructing routes
across the networks that compose the Internet. The Border Gateway
Protocol (BGP), is the \textit{de facto} standard interdomain routing
protocol. As sketched in the example above, it grows a routing tree
towards every destination network in a distributed manner. The example
also shows that the behaviour of the BGP is naturally modelled as a
game, as already pointed out before (see
\cite{FabrikantPapadimitriou,SamiSchapiraZohar} for example). In
particular, checking for so-called \emph{safety} (does the protocol
always converge to a stable state?) amounts to checking termination of
some dynamics. In Section~\ref{sec:applications}, we formally express
BGP in our game model; revisit a classical result of Sami \textit{et
  al.}  that we re-prove within our framework; and finally obtain a
new result regarding BGP: we provide a novel necessary and sufficient
condition for convergence in the restricted (yet realistic) setting
where the preferences of the nodes range on the next-hop in the route
only.


Due to space constraints, full proofs and some examples can be found
in Appendix.


\section{Preliminaries}\label{Sec:Prelim}

\subparagraph{Graphs.} A (directed) \emph{graph} is a pair $G=(V,E)$
where $V$ is a set of \emph{states} (or \emph{nodes}),
$E\subseteq V\times V$ is the set of \emph{edges}. A \emph{labelled
  graph} is a tuple $G=(V,E,L)$ where $(V,E)$ is a graph, and
$L:E\rightarrow S$ is a function associating, to each edge $e$, a
label $L(e)$ from a set $S$ of labels. A~(labelled) graph $G$ is
finite iff $V$ is finite. A \emph{path} in a (labelled) graph $G$ is a
finite sequence $v_1v_2\cdots v_k$ or an infinite sequence
$v_1v_2\cdots v_i\cdots$ of states such that $(v_i,v_{i+1})\in E$ for
all $i$. We denote $v_1$, the first state of a path $\pi$, by
$\first{\pi}$. When $\pi=v_1v_2\cdots v_k$ is finite, we let
$\last{\pi}=v_k$.  We let
$V_\bot=\{v\in V\mid \text{there is no } v': (v,v')\in E\}$ be the set
of \emph{terminal states}. We say that a path $\pi$ is \emph{maximal}
iff: either $\pi$ is infinite, or $\pi$ is finite and
$\last{\pi}\in V_\bot$. Let $\pi_1=v_1\cdots v_k$ and
$\pi_2=u_1u_2\cdots$ be two paths such that $(v_k,u_1)\in E$. Then, we
write $\pi_1\pi_2$ to denote the new path $v_1\cdots v_ku_1u_2\cdots$,
obtained by the concatenation of $\pi_1$ and $\pi_2$.\todo{M: la
  concaténation c'est pas quand $last(\pi_1)=first(\pi_2)?$}

Following automata terminologies, a labelled graph $G$ is said to be
\emph{complete deterministic} if for every state $v$ and label $\ell$,
there is exactly one edge $(v,v')$ s.t. $L(v,v')=\ell$.

\subparagraph{Games played on graphs.} An \emph{$n$-player game} is a
tuple
$\game=(V, E, (V_i)_{1\leq i\leq n}, (\preceq_i)_{1\leq i\leq n})$
where players are denoted by $1,\ldots,n$ and: $(V,E)$ is a finite
graph which forms the \emph{arena} of the game, with $V_\bot$ the
terminal states; $(V_i)_{1\leq i\leq n}$ is a partition of
$V\setminus V_\bot$ indicating which player owns each (non-terminal)
state of the game ($v$ belongs to player $i$ iff $v\in V_i$); and
$\preceq_i$ describes the preference of player $i$ as a reflexive,
transitive and total (i.e.~for all $\pi,\pi'$, $\pi\preceq_i\pi'$ or
$\pi'\preceq_i \pi$) binary relation defined on maximal paths which we
call \emph{plays} (the set of all plays being denoted by $\play$).
Intuitively, player $i$ prefers play $\pi$ to play $\pi'$ iff
$\pi'\preceq_i \pi$. We can extract from $\preceq_i$ a strict partial
order relation by letting $\pi \prec_i \pi'$ if player $i$ strictly
prefers play $\pi'$ to play $\pi$, i.e.~if $\pi\preceq_i\pi'$ and
$\pi'\not\preceq_i \pi$. We also write $\pi\eq_i\pi'$ if
$\pi\preceq_i \pi'$ and $\pi'\preceq_i \pi$, and say that $\pi$ and
$\pi'$ are equivalent for player~$i$. From now on, we describe
preferences by mentioning plays of interest only (implicitly, all
unmentioned plays are equivalent, and below in the preference
order). We also abuse notations and identify a game with its arena:
so, we can write, for instance, about the `paths of $\game$', meaning
the paths of the underlying arena.

\begin{figure}\centering
\begin{tikzpicture}[scale=0.85]
	\draw (0,2) node [state] (j1) {$v_1$};
	\draw (3,2) node [state] (j2) {$v_2$};
	\draw (1.5,0.5) node [state] (cible) {$v_\bot$};

	\draw[->,bend left=10]  (j1) to node[above]{$c_1$} (j2);
	\draw[->,bend left=10]  (j2) to node[below]{$c_2$} (j1);	
	\draw[->]  (j1) to node[left] {$s_1$} (cible);
	\draw[->]  (j2) to node[right] {$s_2$} (cible);	
\end{tikzpicture}\hspace*{10mm}
\begin{tikzpicture}[scale=0.85]
	\draw (0,2) node [state,rectangle] (cc) {$c_1c_2$};
	\draw (2,2) node [state,rectangle] (sc) {$s_1c_2$};
	\draw (0,0.5) node [state,rectangle] (cs) {$c_1s_2$};
	\draw (2,0.5) node [state,rectangle] (ss) {$s_1s_2$};

    \draw[->] (cc) to (sc);
    \draw[->] (cc) to (cs);
    \draw[->] (ss) to (sc);    
    \draw[->] (ss) to (cs);
     
    \end{tikzpicture}\hspace*{10mm}
    \begin{tikzpicture}[scale=0.85]
	\draw (0,2) node [state,rectangle] (cc) {$c_1c_2$};
	\draw (2,2) node [state,rectangle] (sc) {$s_1c_2$};
	\draw (0,0.5) node [state,rectangle] (cs) {$c_1s_2$};
	\draw (2,0.5) node [state,rectangle] (ss) {$s_1s_2$};

    \draw[->] (cc) to (sc);
    \draw[->] (cc) to (cs);
    \draw[->] (ss) to (sc);    
    \draw[->] (ss) to (cs);
    
    \draw[->, bend left=10] (ss) to (cc);    
    \draw[->, bend left=10] (cc) to (ss);    
    \end{tikzpicture}
\caption{Left: a 2-player game $\game^{DIS}$. Middle: $\graph{\game^{DIS}}\POdyn$. Right: $\graph{\game^{DIS}}\PCdyn$}
   \label{Fig:DISAGREE+graphDyn}
 \end{figure}

\begin{example}\label{Ex:DISAGREE}
  Consider the example of \cite{Griffin02}. In our context, it is
  modelled with the 2-player game
  $\game^{DIS}=(V,E,(V_1,V_2), (\preceq_1,\preceq_2))$ depicted on the
  left of \figurename~\ref{Fig:DISAGREE+graphDyn}. The state $v_\bot$
  is terminal. Player $1$ owns $V_1=\{v_1\}$, and player $2$ owns
  $V_2=\{v_2\}$. Let $E=\{c_1,s_1,c_2,s_2\}$ be such that
  $s_i=(v_i,v_\bot)$ and $c_1=(v_1,v_2),\ c_2=(v_2,v_1)$. Edges $c_i$
  stand for `continue', and edges $s_i$ stand for `stop'. For player
  $1$, we let the preferences be
  $(v_1v_2)^{\omega}\prec_1 v_1v_\bot\prec_1 v_1v_2v_\bot$, where
  $\pi^\omega$ denotes an infinite number of iterations of the
  cycle~$\pi$. Symmetrically, player 2 has preferences
  $(v_2v_1)^{\omega}\prec_2v_2v_\bot\prec_2 v_2v_1v_\bot$. In this
  case, all unmentioned plays are equally worse for both players, in
  particular the plays that do not start in the state owned by the
  player (this will always be the case in the routing application of
  Section~\ref{sec:applications}).
\end{example}

\subparagraph{Strategies and strategy profiles.} The game is played by
letting players move a token along the edges of the arena. Note that,
in our games, there is no designated initial state, so the play can
start in any state $v$. The choice of the initial state is not under
the control of any player. Then, the player who owns $v$ picks an
edge $(v,w)$ and moves the token to $w$. It is then the turn of
the player who owns $w$ to choose an edge $(w,u)$ and so
forth. The game continues \emph{ad infinitum} or until a terminal node
has been reached, thereby forming a play. Of course, each player will
act in order to yield a play that is best according to his preference
order $\prec_i$. Since no player controls the choice of the initial
vertex, the players will seek to obtain the best path considering
\emph{any possible initial vertex} (see the formal definitions below).
This will be important for the application of Interdomain Routing in
Section~\ref{sec:applications}, where the games are networks and each
state corresponds to a network node that wants to send a packet to one
of the terminal states.

Formally, a non-maximal path is called a \emph{history} in the
following, and the set of all histories is denoted by $\hist$. We let
$\hist_i$ be the set of histories $h$ such that $\last{h}\in V_i$,
i.e.~$h$ ends in a state that belongs to player $i$. We further let
$\player{h}=i$ iff $h\in \hist_i$. The way players behave in the game
is captured by the central notion of \emph{strategy}, which is a
mapping from a history $h$ to a successor state in the graph,
indicating how the player will play from $h$.  A~\emph{player~$i$
  strategy} is thus a function $\sigma_i\colon\hist_i\to V$ such that,
for all $h\in \hist_i$, $(\last{h},\sigma_i(h))\in
E$. A~\emph{strategy profile} $\sigma$ is a tuple
$(\sigma_i)_{1\leq i\leq n}$ of strategies, one for each player $i$.
In the following, when we consider a strategy profile $\sigma$, we
always assume that $\sigma_i$ is the corresponding strategy of
player~$i$. We also abuse notations, and write $\sigma(h)$ to denote
the node obtained by playing the relevant strategy of $\sigma$ from
$h$, i.e.~$\sigma(h)=\sigma_i(h)$ with $i=\player{h}$. We denote by
$\Sigma_i(\game)$ and $\Sigma(\game)$ the sets of player~$i$
strategies and of strategy profiles respectively (if the game~$\game$
is clear from the context, we may drop it and write $\Sigma$ and
$\Sigma_i$). As usual, given a strategy profile
$\sigma=(\sigma_i)_{1\leq i\leq n}$ and a strategy $\sigma'_j$ for
some player $j$, we denote by $(\sigma_{-j},\sigma_j')$ the strategy
profile obtained from $\sigma$ by replacing the player $j$ strategy
$\sigma_j$ with~$\sigma_j'$.  Fixing a history $h$ (or, in particular,
an initial node) and a profile of strategies $\sigma$ is sufficient to
determine a unique play that is called the \emph{outcome}: we let
$\Outc{\sigma}{h}$ be the (unique) play $hv_1v_2\cdots$ such that for
all $i\geq 1$: $v_i=\sigma(hv_1\cdots v_{i-1})$.

Of particular interest are the \emph{positional} strategies (sometimes
called \emph{memoryless}), i.e.~the set of strategies such that the
action of the player depends on the last state of the history
only. That is, $\sigma_i$ is positional iff for all pairs of histories
$h_1$ and $h_2$ in $\hist_i$: $\last{h_1}=\last{h_2}$ implies
$\sigma_i(h_1)=\sigma_i(h_2)$. For a \emph{positional strategy
  profile} $\sigma$, and a state~$v\in V$, we write $\sigma(v)$ to
denote the (unique) state $\sigma(h)$ returned by $\sigma$ for all $h$
with $\last{h}=v$. We denote by $\pstrat(\game)$ the set of strategy
profiles composed of positional strategies only, and by
$\pstrat_i(\game)$ the set of player $i$ positional strategies. From
all states $v$, applying a positional strategy profile builds a play
such that the very same decision is always taken at a particular
state: therefore, it either creates a finite path without cycles, or a
lasso (infinite path that starts with a finite path without cycle and
continues with an infinite simple cycle, disjoint from the finite
path). We let $\playP$ be the set of all \emph{positional plays} thus
generated. In a game where we are only interested in positional
strategies (as this will be the case in the application to routing,
for instance), the preferences need only be defined on positional
plays. Indeed, all other plays will never be obtained as an outcome,
and can be assumed to be worse than any other positional play.

\subparagraph{Game Dynamics.} Let us now turn our attention to the
central notion of \emph{dynamics}.  Intuitively, a dynamics consists
in letting players update their strategies according to some
criteria. For example, a player will want to update his strategy in
order to yield a better outcome according to his
preferences. Therefore, a dynamics can be understood as a graph whose
states are the strategy profiles and whose edges correspond to
possible updates.

\begin{definition}
  Let $\game$ be a game. A dynamics for $\game$ is a binary relation
  ${\to}\subseteq\Sigma\times\Sigma$ over the strategy profiles of
  $\game$. Its associated graph is $\graph\game\to=(\Sigma,{\to})$,
  where $\Sigma$ is the set of states.  The terminal profiles $\sigma$
  of $\graph\game\to$ (without outgoing edges) are
  called the \emph{equilibria} of $\to$.
\end{definition}

We will focus on five dynamics, modelling certain rational behaviours
of the players:
\begin{itemize}
\item The \emph{one-step} dynamics $\Odyn$. It corresponds to the
  minimal update that can occur, where only one player changes a
  single decision in order to improve the outcome from his point of
  view: $\sigma\Odyn\sigma'$ iff there is a player
  $i\in\{1,\ldots,n\}$ and a history $h\in \hist_i$ such that
  (i) $\sigma(h)\neq\sigma'(h)$; 
  (ii) $\Outc{\sigma}{h}\prec_i\Outc{\sigma'}{h}$; and
  (iii) $\sigma(h')=\sigma'(h')$ for all $h'\neq h$.
   Note that the equilibria of the one-step dynamics are
  exactly the so-called \emph{subgame perfect equilibria} (SPE)
  introduced in \cite{Se65} (see also \cite{OR94}).
\item The \emph{positional one-step} dynamics $\POdyn$. It ranges over
  positional strategy profiles only, and corresponds to a single
  player updating his strategy from a single state. Formally,
  $\sigma\POdyn\sigma'$ (with $\sigma,\sigma'\in \pstrat$) iff there
  are a player $i\in\{1,\ldots,n\}$ and a state $v\in V_i$ s.t.
  (i) $\sigma(v)\neq\sigma'(v)$;
  (ii) $\Outc{\sigma}{v}\prec_i\Outc{\sigma'}{v}$; and
  (iii) $\sigma(v')=\sigma'(v')$ for all $v'\neq v$.
  
\item The \emph{best reply positional one-step dynamics} $\BPOdyn$. We
  let $\sigma \BPOdyn \sigma'$ iff there exists a player
  $i\in\{1,\ldots,n\}$ and a state $v\in V_i$ such that the three
  properties of the positional one-step dynamics are satisfied, and,
  in addition, the following \emph{best-reply} condition is satisfied:
  (iv) for all $\sigma''\neq \sigma'$ such that
    $\sigma\POdyn\sigma''$ \emph{if} player~$i$ is the one that has
    changed its strategy between $\sigma$ and $\sigma''$, then:
    $\Outc{\sigma''}{v}\preceq_i\Outc{\sigma'}{v}$.

\item The \emph{positional concurrent} dynamics $\PCdyn$ and its best
  reply version $\BPCdyn$. Several players can update their strategies
  at the same time (in a `one step' fashion),
  but each individual update would yield a better play when performed
  independently (in some sense, each player performing an update
  `believes' he will improve). Formally, for
  $\sigma,\sigma'\in \pstrat$, we let $\sigma\PCdyn\sigma'$
  (respectively, $\sigma\BPCdyn\sigma'$) iff for all
  $i\in P(\sigma, \sigma')$, $\sigma\POdyn (\sigma'_i,\sigma_{-i})$
  (respectively, $\sigma\BPOdyn (\sigma'_i,\sigma_{-i})$).
 
\end{itemize}

\todo{M: On ne fait aucune référence à notre article Gandalf, alors
  que c'est quand même la continuité. On ne dirait pas un mot ici? Par
  exemple que d'autres types de dynamiques on été étudiées sur les
  jeux sur arbre ou quelque chose comme ça? Je sais pas si c'est une
  bonne idée, je propose...}  Observe that other dynamics can be
defined, corresponding to other behaviours of the players. We focus on
these five dynamics as they fit the applications we target in
Section~\ref{sec:applications}. We have already said that the
equilibria of $\Odyn$ are SPEs, and we can also see from the
definitions that the equilibria of the four other dynamics coincide.

\begin{example}
  Let $\game^{\mathrm{DIS}}$ be the game from Example~\ref{Ex:DISAGREE}. The
  graphs $\graph{\game^{\mathrm{DIS}}}\POdyn$ and $\graph{\game^{\mathrm{DIS}}}\PCdyn$
  are given in the middle and the right of
  \figurename~\ref{Fig:DISAGREE+graphDyn}, where each strategy profile
  is represented by the choices of the players from $v_1$ and
  $v_2$. For example, $c_1c_2$ is the strategy profile
  s.t. $\sigma_1(v_1)=v_2$ and $\sigma_2(v_2)=v_1$. Note that, in this
  example, ${\POdyn}={\BPOdyn}$ and ${\PCdyn}={\BPCdyn}$. Moreover, we can see that $\graph{\game^{\mathrm{DIS}}}\POdyn$ has no infinite paths, contrary to $\graph{\game^{\mathrm{DIS}}}\PCdyn$. We then say that the dynamics $\POdyn$ terminates on $\game^{\mathrm{DIS}}$, while $\PCdyn$ does not terminate on $\game^{\mathrm{DIS}}$.
\end{example}

The main problem we study is whether a given dynamics terminates on a
certain game: we say that a dynamics $\to$ terminates on the game
$\game$ if there is no infinite path in the graph $\graph\game\to$ of
the dynamics. As illustrated in the introduction
(Example~\ref{Ex:DISAGREE}), such infinite paths may be problematic in
certain applications, like in the Interdomain Routing problem, where
an infinite path in the dynamics means that the routing protocol does
not stabilise. We are thus interested in techniques to check whether a
dynamics terminates on a given game.

Sometimes, a dynamics does not terminate in general, but does when we
restrict ourselves to \emph{fair executions} where all players will
eventually have the opportunity to update their strategies if they
want~to. Formally, given a dynamics $\to$, an infinite path
$\sigma^1\to \sigma^2\to\cdots$ of the graph $\graph\game\to$ is
\emph{not fair} if there exists a player $i$, and a position $k$ such
that for all $\ell\geq k$, player $i$ can switch his strategy
in~$\sigma^\ell$ (i.e.~there is $\sigma^\ell\to\sigma'$ where
$\sigma^\ell_i\neq\sigma'_i$), but for all $\ell\ge k$, player $i$
keeps the same strategy forever (i.e.~$\sigma^\ell_i=\sigma^k_i$). We
say that the dynamics $\to$ \emph{fairly terminates} for the game
$\game$ if there are no infinite fair paths in the graph
$\graph\game\to$: this is a weakening of the notion of termination
seen before (see \cite{techrep} for an example of a dynamics that
does not terminate but terminates fairly).


\section{Simulations: preorders on the dynamics graphs}\label{Sec:Graphs}
At this point of the paper, it is important to understand that a game
is characterised by two graphs: the \emph{game graph} which gives its
\emph{structure} (see for example, \figurename~\ref{Ex:DISAGREE},
left); and the \emph{dynamics graph}, which, given a fixed dynamics
$\to$, defines the semantics of the game as the long-term behaviour of
the players (\figurename~\ref{Fig:DISAGREE+graphDyn}, middle and
right). In the present section, we study \emph{preorder relations} on
the \emph{dynamics graphs}, relying on the classical notion of
\emph{simulation} \cite{Milner}. They are the key ingredients to
reason about the termination of dynamics.

The domain of a binary relation $R\subseteq A\times B$ is the set of
elements $a\in A$ such that there exists $b\in B$ with $(a,b)\in
R$. The co-domain or $R$ is the set of elements $b\in B$ such that
there exists $a\in A$ with $(a,b)\in R$. We denote the domain of $R$
by $\dom(R)$. The transitive closure $R^+$ of relation $R$ is defined
as $(a,b)\in R^+$ iff there are $a_0=a, a_1,a_2,\ldots,a_n=b$ such
that for all $i\in \{0,1,\ldots,n-1\}$, $(a_i,a_{i+1})\in R$.

\subparagraph{Partial simulations and simulations.} We start with some
weak version of the notion of simulation, called \emph{partial
  simulation} $\sqsubseteq$. Intuitively, we say that a state $u$
\emph{partially simulates} a state $u'$ (noted $u'\sqsubseteq u$) if
for all successor states $v'$ of $u'$, the following holds: \emph{if
  $v'$ is in the domain of the simulation}, then there must be some
state $v$ simulating $v'$ such that $v$ is a successor of $u$.
Formally, if $G=(V,E)$ and $G'=(V',E')$ are two graphs, a binary
relation $\sqsubseteq$ contained in $V'\times V$ is a \emph{partial
  simulation} of $G'$ by $G$ if: for all
$(u',v')\in E'\cap\dom(\sqsubseteq)^2$, for all $u\in V$:
$u'\sqsubseteq u$ \emph{implies} there is $v\in V$ such that
$(u,v)\in E$ and $v'\sqsubseteq v$.  Then, a \emph{simulation}
$\sqsubseteq$ of $G'$ by $G$ is a partial simulation of $G'$ by $G$
s.t. $\dom(\sqsubseteq)=V'$, i.e. all states of $G'$ are simulated by
some state of $G$. When a (partial) simulation $\sqsubseteq$ of $G'$
by $G$ exists, we say that $G$ (partially) simulates $G'$.  The
following example highlights the difference between partial
simulations and simulations. Assume $G$ with only one edge $u\to v$
and $G'$ with only two edges $u'\to v_1'$ and $u'\to v_2'$. Then, the
relation $\sqsubseteq$ s.t. $u'\sqsubseteq u$ and $v_1'\sqsubseteq v$
(but $v_2'\not\sqsubseteq v$) is a partial simulation (its domain is
$\{u',v_1'\}$ so it is not a problem that $v_2'$ is not simulated) but
is not a simulation relation.

Simulations between dynamics graphs help in showing termination
properties, as shown by the following folk result:
\begin{proposition}\label{Prop:simulation=>same termination}
  Let $\game_1$ and $\game_2$ be two games, $\to_1$ and $\to_2$ be two
  dynamics on $\game_1$ and $\game_2$ respectively. If
  $\graph{\game_1}{\to_1}$ simulates $\graph{\game_2}{\to_2}$ and the
  dynamics $\to_1$ terminates on $\game_1$, then the dynamics $\to_2$
  terminates on $\game_2$.
\end{proposition}

\subparagraph{Bisimulations and transitive closure.} We can define other
preorder relations on dynamics graphs. A bisimulation is a simulation
$\sqsubseteq$ such that the inverse relation $\sqsubseteq^{-1}$ is
also a simulation. We say that $G=(V,E)$ and $G'=(V',E')$ are
bisimilar when there is a bisimulation between them.  As a corollary
of the previous proposition, if $\graph{\game_1}{\to_1}$ and
$\graph{\game_2}{\to_2}$ are bisimilar, then $\to_1$ terminates on
$\game_1$ if and only if $\to_2$ terminates on $\game_2$.

For termination purposes, it is also perfectly fine to simulate a
single step of $G'$ in several steps of $G$ for instance. The
following proposition stems from Proposition~\ref{Prop:simulation=>same
  termination} and mixes the notions of transitive closures and
partial simulations.

\begin{proposition}\label{Cor:weak et co-weak=>same termination} Let
  $\game_1$ and $\game_2$ be two games, $\to_1$ and $\to_2$ be
  dynamics on $\game_1$ and $\game_2$ resp.
  \begin{itemize}
  \item If $\graph{\game_1}{\to_1^+}$ simulates
    $\graph{\game_2}{\to_2}$ and the dynamics $\to_1$ terminates on
    $\game_1$, then the dynamics $\to_2$ terminates on $\game_2$.
  \item If $\sqsubseteq$ is a partial simulation of
    $\graph{\game_2}{\to_2^+}$ by $\graph{\game_1}{\to_1^+}$, and the
    dynamics $\to_1$ terminates on~$\game_1$, then there are no paths
    in $\graph{\game_2}{\to_2}$ that visit a state of
    $\dom(\sqsubseteq)$ infinitely often.
  \end{itemize}
\end{proposition}


\section{Minors and domination: preorders on game graphs}\label{Sec:Game}

Let us now introduce notions of \emph{preorders on game graphs}. We
introduce a new notion of graph minor which consists in lifting the
classical notion of graph minor to the context of $n$-player games on graphs. To
the best of our knowledge, this has not been done previously. This new
preorder on game graphs enables us to use in a simple context the results
of Section~\ref{Sec:Graphs} to reason about termination of
dynamics. Let us start with the formal definition. For that purpose,
we start by defining two transformations on game graphs. Let
$\game=(V, E, (V_i)_{i}, (\preceq_i)_{i})$ be an $n$-player game. Then
we can modify it by applying either of the following transformations
that yields a game $\game'=(V', E', (V'_i)_{i}, (\preceq'_i)_{i})$.
\begin{itemize}
\item \emph{Deletion of an edge} $(u,v)\in E$. Then, $V'=V$,
  $E'=E\setminus \{(u,v)\}$, $(V'_i)_{i} = (V_i)_{i} $ and
  $\preceq'_i$ is s.t. $\pi_1\preceq'_i\pi_2$ iff
  $\pi_1\preceq_i\pi_2$ and $\pi_1$, $\pi_2$ are both paths of
  $\game'$.
\item \emph{Deletion of a state} $v\in V_j$ (for a certain player
  $j$). This can happen in two different ways:
  \begin{enumerate}
  \item either when $v$ is isolated, i.e. when $(u,v)\not\in E$ and
    $(v,u)\not\in E$ for all $u\in V$. Then, $V'=V\setminus\{v\}$,
    $E'=E$, $V'_i = V_i$ for all $i\neq j$, $V'_j=V_j\setminus\{v\}$,
    and $(\preceq'_i)_{i}=(\preceq_i)_{i}$.
  \item or when $v$ has a unique outgoing edge $(v,v')$ and all
    predecessors $u$ of $v$ (i.e. $(u,v)\in E$) do not have $v'$ as a
    successor (i.e.~$(u,v')\notin E$). In this case, we have
    $V'=V\setminus\{v\}$, $V'_i=V_i$ for all $i\neq j$ and
    $V'_{j} = V_{j}\setminus \{v\}$,
    $E'=\left(E\cap (V'\times V')\right)\cup \{(u,v')\mid (u,v)\in E\}$, and
    $\pi_1'\preceq'_i \pi_2'$ iff $\pi_1\preceq_i\pi_2$ where $\pi_1$
    and $\pi_2$ are the plays of $\game$ obtained from $\pi_1'$ and
    $\pi_2'$ respectively, by replacing all occurrences of $(u,v')$
    (for some $u$) by $(u,v),(v,v')$.
  \end{enumerate}
\end{itemize}

\begin{definition}
  Let $\game$ and $\game'$ be two $n$-player games. Then, $\game'$ is
  a minor of $\game$ if $\game'$ can be obtained from $\game$ by
  applying a sequence of edges and states deletions.
\end{definition}

\begin{example}
  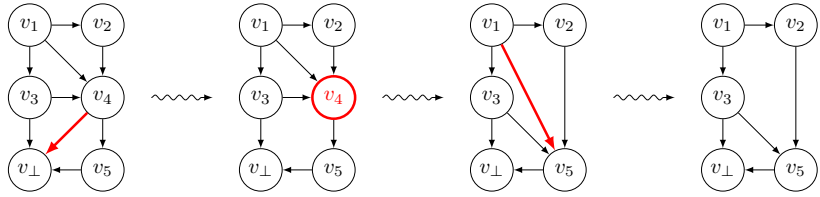
\begin{figure}[tbp]
  \centering
  \scalebox{.8}{\begin{tikzpicture}[node distance=1.2cm][scale=0.85]
    \node[state] (v1){$v_1$};
    \node[state,right of=v1] (v2){$v_2$};
    \node[state,below of=v1] (v3){$v_3$};
    \node[state,right of=v3] (v4){$v_4$};
    \node[state,below of=v3] (vbot){$v_\bot$};
    \node[state,right of=vbot] (v5){$v_5$};

    \path[->]
    (v1) edge (v2) edge (v4) edge (v3)
    (v2) edge (v4)
    (v3) edge (v4) edge (vbot)
    (v4) edge (v5) edge[color=red,very thick] (vbot)
    (v5) edge (vbot);
    
    \draw[->,decorate,decoration={snake,amplitude=.4mm,segment
      length=2mm,post length=1mm}] (2,-1.2) -- (3.01,-1.2);

    \begin{scope}[xshift=3.8cm]
      \node[state] (v1){$v_1$};
      \node[state,right of=v1] (v2){$v_2$};
      \node[state,below of=v1] (v3){$v_3$};
      \node[state,right of=v3,color=red,very thick] (v4){$v_4$};
      \node[state,below of=v3] (vbot){$v_\bot$};
      \node[state,right of=vbot] (v5){$v_5$};

      \path[->]
      (v1) edge (v2) edge (v4) edge (v3)
      (v2) edge (v4)
      (v3) edge (v4) edge (vbot)
      (v4) edge (v5)
      (v5) edge (vbot);
      
      \draw[->,decorate,decoration={snake,amplitude=.4mm,segment
        length=2mm,post length=1mm}] (2,-1.2) -- (3.01,-1.2);
    \end{scope}

    \begin{scope}[xshift=7.6cm]
      \node[state] (v1){$v_1$};
      \node[state,right of=v1] (v2){$v_2$};
      \node[state,below of=v1] (v3){$v_3$};
      \node[state,below of=v3] (vbot){$v_\bot$};
      \node[state,right of=vbot] (v5){$v_5$};

      \path[->]
      (v1) edge (v2) edge[color=red,very thick] (v5) edge (v3)
      (v2) edge (v5)
      (v3) edge (v5) edge (vbot)
      (v5) edge (vbot);
      
      \draw[->,decorate,decoration={snake,amplitude=.4mm,segment
        length=2mm,post length=1mm}] (2,-1.2) -- (3.01,-1.2);
    \end{scope}

    \begin{scope}[xshift=11.4cm]
      \node[state] (v1){$v_1$};
      \node[state,right of=v1] (v2){$v_2$};
      \node[state,below of=v1] (v3){$v_3$};
      \node[state,below of=v3] (vbot){$v_\bot$};
      \node[state,right of=vbot] (v5){$v_5$};

      \path[->]
      (v1) edge (v2) edge (v3)
      (v2) edge (v5)
      (v3) edge (v5) edge (vbot)
      (v5) edge (vbot);
    \end{scope}
  \end{tikzpicture}}
  \caption{Minors obtained by first deleting the edge $(v_4,v_\bot)$,
    then the state $v_4$ (that has now a unique successor $v_5$), and
    then the edge $(v_1,v_5)$}
  \label{fig:minor}
\end{figure}
  An example of minor is depicted in \figurename~\ref{fig:minor}. If
  the original preferences of the player owning state $v_1$ are
  $v_1v_4v_5v_\bot \prec v_1v_3v_\bot \prec v_1v_2v_4v_\bot\prec
  v_1v_2v_4v_5v_\bot$ (other plays being equally worse for this
  player), then after the deletion of the edge $(v_4,v_\bot)$, his
  preferences become
  $v_1v_4v_5v_\bot\prec v_1v_3v_\bot\prec v_1v_2v_4v_5v_\bot$ (the
  path $v_1v_2v_4v_\bot$ does not exist in the new graph and has
  simply been removed from the preferences). Next, the deletion of
  $v_4$ is allowed because it is a single outgoing edge $v_5$, and
  neither $v_1$ nor $v_2$ nor $v_3$ have an edge to $v_5$. After this
  deletion, the preferences become
  $v_1v_5v_\bot \prec v_1v_3v_\bot\prec v_1v_2v_5v_\bot$. Finally,
  after the deletion of the edge $(v_1,v_5)$, the preferences become
  $v_1v_3v_\bot\prec v_1v_2v_5v_\bot$.
\end{example}

The deletion of a state therefore consists in squeezing each path of
length $2$ around it in a single edge. In the example, the deletion of
the state $v_4$ consists in squeezing the paths $v_1v_4v_5$ in the
edge $v_1v_5$, and the same for $v_2v_4v_5$ and $v_3v_4v_5$ in the
edges $v_2v_5$ and $v_3v_5$ respectively. The condition
$(u,v')\notin E$ makes sure that this squeezing is not perturbed by
the presence of an incident edge $(u,v')$ that could have
contradictory preferences. For instance, in the previous example, we
cannot remove vertex $v_2$ in the minor obtained before having removed
edge $(v_1,v_5)$: otherwise, we would obtain as preferences for the
owner of $v_1$ the chain
$v_1v_5v_\bot\prec v_1v_3v_\bot\prec v_1v_5v_\bot$ which is not
possible.

We can link termination of dynamics on graph games to the presence of
minors, in the various dynamics introduced before: if we manage to
find a game minor where the dynamics does not terminate, then the
original game does not terminate either.

\begin{theorem}\label{Thm:minor of game=>simulation of graph}
  Let $\game$ be a game, and $\game'$ be a minor of $\game$. If
  ${\to}\in\{\Odyn,\POdyn,\PCdyn\}$, then $\graph\game\to$ simulates
  $\graph{\game'}\to$. In particular, via
  Proposition~\ref{Prop:simulation=>same termination}, if the dynamics
  $\to$ 
  terminates for $\game$, then it terminates for $\game'$ too.
\end{theorem}
\begin{proof}[Sketch of proof]
  We prove the result for $\Odyn$ ; the two other cases are
  similar. Since simulations are transitive relations, it is
  sufficient to only consider that $\game'$ has been obtained from
  $\game$ either by deleting a single edge, or by deleting a single
  node. Let us briefly detail the case where $\game'$ is obtained by
  the deletion of a state $v$. If
  $h\in\hist(\game)\setminus \{h\mid last(h)=v\}$, we
  can construct a corresponding play $f(h)$ of $\game'$ by replacing a
  sequence $uvv'$ of $h$ by $uv'$. The conditions over the deletion of
  $v$ implies that that $f$ is indeed a bijection from
  $\hist(\game)\setminus \{h\mid last(h)=v\}$ to
  $\hist(\game')$.  We then consider the following relation on
  strategy profiles: $\sigma'\sqsubseteq \sigma$ if for all histories
  $h\in \hist(\game)\setminus \{h\mid last(h)=v\}$,
  $\sigma'(f(h))=\sigma(h)$ if $\sigma(h)\neq v$, and
  $\sigma'(f(h))=v'$ otherwise; and show that $\sqsubseteq$ is a
  simulation (indeed a bisimulation).
\end{proof}

\begin{figure}[tbp]\centering
\begin{tikzpicture}[scale=0.8]
	\draw (0,2) node [state] (j1) {$v_1$};
	\draw (3,2) node [state] (j2) {$v_2$};
	\draw (1.5,0.5) node [state] (cible) {$v_\bot$};
        \draw (0,0.5) node[state, very thick] (j3) {$v_3$} ;

	\draw[->,bend left=10]  (j1) to node[above]{$c_1$} (j2);
	\draw[->,bend left=10]  (j2) to node[below]{$c_2$} (j1);	
	\draw[->]  (j1) to node[left] {$s_1$} (cible);
	\draw[->]  (j2) to node[right] {$s_2$} (cible);
        \draw[->,very thick]  (j1) to node[left] {$d$} (j3);
        \draw[->,very thick]  (j3) to (cible);

      \end{tikzpicture}\hspace*{5mm}
\scalebox{.9}{\begin{tikzpicture}[scale=0.8]
	\draw (0,2) node [state,rectangle] (cc) {$c_1c_2$};
	\draw (2,2) node [state,rectangle] (sc) {$s_1c_2$};
        \draw (4,2) node [state,rectangle] (dc) {$dc_2$};
	\draw (0,0.5) node [state,rectangle] (cs) {$c_1s_2$};
	\draw (2,0.5) node [state,rectangle] (ss) {$s_1s_2$};
        \draw (4,0.5) node [state,rectangle] (ds) {$ds_2$};

        \path[->]
        (cc) edge[bend left] (dc)
             edge (ds)
             edge (cs)
             edge (sc)
             edge[bend left=10] (ss)
        (sc) edge (dc)
        (ss) edge (sc)
             edge (dc)
             edge (ds)
             edge (cs)
             edge[bend left=10] (cc)
        (dc) edge (ds)
        ; 
      \end{tikzpicture}\hspace*{5mm}
\begin{tikzpicture}[scale=0.85]
	\draw (0,2) node [state,rectangle] (cc) {$c_1c_2$};
	\draw (2,2) node [state,rectangle] (sc) {$s_1c_2$};
        \draw (4,2) node [state,rectangle] (dc) {$dc_2$};
	\draw (0,0.5) node [state,rectangle] (cs) {$c_1s_2$};
	\draw (2,0.5) node [state,rectangle] (ss) {$s_1s_2$};
        \draw (4,0.5) node [state,rectangle] (ds) {$ds_2$};

        \path[->]
        (cc) edge[bend left] (dc)
             edge (ds)
             edge (cs)
        (sc) edge (dc)
        (ss) edge (sc)
             edge (dc)
             edge (ds)
        (dc) edge (ds)
        ;     
      \end{tikzpicture}}
      
    \caption{L: a 3-player game $\game$ with $\game^{\mathrm{DIS}}$
      (\figurename~\ref{Fig:DISAGREE+graphDyn}) as a minor. M:
      $\graph\game\PCdyn$. R: $\graph{\game}\BPCdyn$}
   \label{Fig:ce-theorem-9}
 \end{figure}

Notice that Theorem~\ref{Thm:minor of game=>simulation of graph}
suffers from three weaknesses. First, it does not hold for the best
reply dynamics $\BPOdyn$ and $\BPCdyn$, as shown by the following
example. Consider again the game $\game^{\mathrm{DIS}}$ from
Example~\ref{Ex:DISAGREE}. Further, consider the game $\game$ 
in \figurename~\ref{Fig:ce-theorem-9} obtained from $\game^{\mathrm{DIS}}$ by adding
a third player, who owns a single node $v_3$, such that the only edges to
and from $v_3$ are $(v_1,v_3)$ and $(v_3,v_\bot)$, and where the
preferences of player~1 are now
$v_1v_\bot\prec_1 v_1v_2v_\bot\prec_1 v_1v_3v_\bot$ (observe that now,
he prefers a path that traverses the new node $v_3$ above all other
paths). Clearly, $\game^{\mathrm{DIS}}$ is a minor of $\game$. Using
Theorem~\ref{Thm:minor of game=>simulation of graph}, and since we
know that $\graph{\game^{\mathrm{DIS}}}{\PCdyn}$ does not terminate, we deduce
that $\graph{\game}{\PCdyn}$ does not terminate either. Moreover, in
this example,
$\graph{\game^{\mathrm{DIS}}}{\PCdyn}=\graph{\game^{\mathrm{DIS}}}{\BPCdyn}$, so even
with the best-response property, the dynamics does not terminate in
the minor. However, one can check that $\graph{\game}{\BPCdyn}$
terminates thanks to the best-response property: Player~1 will not try
to obtain path $v_1v_2v_\bot$ (which leads to a cycle in
$\graph{\game^{\mathrm{DIS}}}{\PCdyn}$), but will choose a strategy going to
$v_3$ (see \figurename~\ref{Fig:ce-theorem-9}, Right). So, $\game^{\mathrm{DIS}}$ is a
minor of $\game$, s.t. $\BPCdyn$ terminates for $\game^{\mathrm{DIS}}$ but not
in $\game$. The example can be adapted to $\BPOdyn$.

A second weakness is that Theorem~\ref{Thm:minor of game=>simulation
  of graph} does not apply to fair termination: the dynamics $\to$
could fairly terminate for the game $\game$, but not for his minor $\game'$. This could
be the case if we remove every choice (except one) for a certain
player in the minor $\game'$ creating a fair cycle in $\game'$ that
would not be present in $\game$. 

Finally, the reciprocal of Theorem~\ref{Thm:minor of game=>simulation
  of graph} does not hold: all dynamics terminate on the trivial graph
with a single state, but it is also minor of all games, including
those where the dynamics does not terminate.

This motivates the introduction of a stronger notion of graph minor,
where it is allowed to remove \emph{only} the so-called
\emph{dominated} edges. Formally, let $\game$ be a game, let
$v\in V_i$ be a state, and let $e_1=(v,v_1)$ and $e_2=(v,v_2)$ be two
outgoing edges of $v$. We say that\emph{ $e_1$ is dominated by $e_2$}
if for all \emph{positional} strategies\footnote{We restrict our
  definition to the context of positional strategies, for the sake of
  brevity, but it can be extended to the more general setting.}
$\sigma\in \pstrat$, $\Outc{\sigma_1}{v}\prec_i\Outc{\sigma_2}{v}$,
where $\sigma_1$ and $\sigma_2$ coincide with $\sigma$ except that
$\sigma_1(v)=v_1$ and $\sigma_2(v)=v_2$. Intuitively, this means that
the player always prefers $e_2$ to $e_1$.  Then, a game $\game'$ is
said to be a \emph{dominant minor} of $\game$ if it can be obtained
from $\game$ by deleting states as before, but only deleting
\emph{dominated} edges. Equipped with this notion, we overcome the
three limitations of Theorem~\ref{Thm:minor of game=>simulation of
  graph} we had identified:

\begin{theorem}\label{theo:fair-termination}\label{Thm:dominant minor}
  Let $\game$ be a game and $\game'$ be a dominant minor of
  $\game$. If ${\to} \in \{\BPOdyn,\BPCdyn\}$, then we can build a
  simulation $\sqsubseteq$ of $\graph{\game'}{\to}$ by
  $\graph{\game}{\to}$ such that:
(i) $\sqsubseteq^{-1}$ is a partial simulation of
    $\graph{\game}{\to}$ by $\graph{\game'}{\to}$; and
  (ii) if there is a fair
  cycle in $\game$ then there is a fair cycle in $\game'$.

  In particular, the dynamics $\to$ fairly terminates for $\game$ if
  and only if it does for $\game'$.
\end{theorem}

Now, Theorem~\ref{theo:fair-termination} has some limitations too. We
can show that it does not hold for the `non-best-reply dynamics'
$\POdyn$ and $\PCdyn$. Moreover, even when we consider best-reply
dynamics, the fairness condition remains crucial: we can exhibit (see
\cite{techrep}) a case where there is a (non-fair) cycle in $\game$
but no cycles in $\game'$.


\section{Applications to interdomain routing convergence}
\label{sec:applications}

As explained in the introduction, the Border Gateway Protocol (BGP) is
the \textit{de facto} standard interdomain routing protocol. Its role
is to establish routes to all the networks that compose the
Internet. BGP does this by growing a routing
tree towards every destination network in a distributed manner, as follows. \todo{M: Si je ne me trompe pas, ce n'est pas exactement un arbre qu'on décrit après, si?} In the initial
state, only the router in the destination network has a route towards
itself that it advertises to its neighbours. Each time a router
receives an advertisement, it selects among the neighbour routes the
one it considers best and then advertises it to its neighbours. The
process repeats until no router wants to change its best route. To
select its best route, a router first filters the received routes to
retain only \emph{permitted one}s and ranks them according to its
\emph{preference}. Both the filtering and ranking of routes by a
router are decided based on the network's routing policy. For
example, a route can be preferred over another because it offers
better performance or costs less and it can be filtered out because it
is not economically viable.

As shown in the introductory example, the routing approach at the
heart of BGP has known convergence issues. It could fail to reach an
equilibrium, entering a persistent oscillatory behaviour or it could
have no equilibrium at all. This is a well-studied problem that has
led to considerable work
\cite{Griffin02,GaoRexford,FabrikantPapadimitriou,SamiSchapiraZohar,CittadiniDR,Griffin18,JaggardLutzSchapiraWright,NRTV07}.
In their seminal work \cite{Griffin02}, Griffin \textit{et al.}
analysed the BGP convergence properties using a simplified model named
the Stable Path Problem (SPP). The main questions they ask are the
following:
(1) whether an SPP instance is \textbf{Solvable}, i.e., whether
  it admits a stable state;
(2) whether the stable state is \textbf{Unique}; and
(3) whether the system is \textbf{Safe}, i.e. it always converges
  to a stable state.

They also give a \emph{sufficient} condition for an SPP instance to be
safe: the absence of a substructure named a \textbf{Dispute
  Wheel}. Later, Sami \textit{et al.} \cite{SamiSchapiraZohar} have
shown that the existence of multiple stable states is a sufficient
condition to prevent safety (i.e. Safe $\implies$ Unique). These
results have later been refined by Cittadini \textit{et al.}
\cite{CittadiniDR}.  While the works just cited focus on the
definition of sufficient conditions for safety, another approach by
Gao and Rexford \cite{GaoRexford} achieves convergence by enforcing
only local conditions on route preferences.

In this section, we show how SPP can be expressed in our $n$-player
game model, therefore \textbf{Safety} reduces to checking for termination
of the game dynamics. We revisit the result of Sami \textit{et al.} by
providing a new proof that relies on our framework. Then, we further
exploit this framework to obtain a new result about SPP: we provide a
\emph{necessary and sufficient} condition for safety in a setting
which is more restricted (yet still realistic) than in \cite{Griffin02}.

\subparagraph{One target games.} We first translate the SPP, as a
combination of:
(1) a reachability game that models the network topology and routing
  policies; and
(2) the best-reply positional concurrent dynamics
  that models the asynchronous behaviour of the routing protocol.

Using this approach, the routing safety problem translates to a
dynamics termination problem.

We rely on a particular class of games, that we call \emph{one target
  games} (\OTG for short): they have a unique target, the destination
network, that all players want to reach. Each player corresponds to a
network in the Internet and as such owns a single state. The routing
policies of networks are modelled by the preference relations and by
the distinction between permitted and forbidden paths. The preferences
are only over positional strategies (paths), meaning that each network
picks its next-hop independently of its predecessors. Permitted and
forbidden paths model the fact that only some paths are allowed by the
networks routing policies. Forbidden paths are also used to take into
account additional restrictions that cannot be directly modelled. In
SPP, the paths are simple (no loops); and non-simple paths are
forbidden, for obvious reasons of efficiency. Moreover, in SPP, if at
some point a network reaches a forbidden path, he will inform his
neighbours that he is not able to reach the target. To model this, we
impose a requirement that that if a path is permitted, all its
suffixes are also permitted.

Formally, let
$\game=(V, E, (V_i)_{1\leq i\leq n}, (\preceq_i)_{1\leq i\leq n})$ be an
$n$-player game. For all $1\leq i\leq n$, we assume that $\perm_i$ is
the set of \emph{permitted paths} of player $i$. All these paths are
finite paths of the form $v_i\cdots \target \in \perm_i$. We denote by
$\nperm_i$ the set of \emph{forbidden paths}, i.e. all the positional
plays starting in $v_i$ that are not in $\perm_i$ (in particular, all
infinite paths are forbidden). We let
$\perm=\bigcup_{1\le i\le n} \perm_i$ and
$\nperm=\bigcup_{1\le i\le n} \nperm_i$.  Then, $\game$ is a \emph{one
  target game} (\OTG) if:
  \begin{itemize}
  \item $\Target=\{\target\}$, and, for all players $i$:
    $V_i=\{v_i\}$;
  \item for all $\pi_1\in \nperm_i$, for all $\pi_2\in \perm_i$:
    $\pi_1\prec_i \pi_2$ (permitted is better than forbidden);
  \item for all $\pi_1,\pi_2\in \nperm_i$: $\pi_1\eq_i \pi_2$ (all
    forbidden paths are equivalent);
  \item for all $\pi_1,\pi_2\in \perm_i$: $\pi_1\eq_i \pi_2$
    \emph{implies} that then there are $v\in V$ and
    $\widetilde{\pi}_1, \widetilde{\pi}_2$ s.t.
    $\pi_1=v_iv\widetilde\pi_1$ and $\pi_2=v_iv\widetilde\pi_2$ (if
    two permitted paths are equivalent, they have the same next-hop);
  \item for all $\pi\in \perm_i$, for all suffixes $\widetilde{\pi}$
    of $\pi$: $\widetilde{\pi}\in \perm$ (all suffixes of permitted
    paths are permitted).
\end{itemize}

Our running example (\figurename~\ref{Fig:DISAGREE+graphDyn}) is a
\OTG. Since, in such a game, each player owns one and only one state, we
will abuse notation by confusing each state $v\in V$ with its
player. For example, for $v\in V_i$, we will write $\prec_v$ instead
of $\prec_i$.

\subparagraph{Sami \textit{et al}: Termination implies a unique
  terminal node.}  Equipped with this definition, we start by revisiting
a result of Sami \textit{et al.} saying that when an instance of SPP
is \emph{safe}, the solution is unique. In our setting, this
translates as follows:
\begin{theorem}\label{thm:one-equilibrium}
  Let $\game$ be a \OTG.  If ${\BPCdyn}$ fairly terminates for $\game$
  (i.e.~the corresponding instance of SPP is safe), then it has
  exactly one equilibrium.
\end{theorem}

We (re-)prove this result in our setting. We rely on the notion of
\Lfair\ path that we define now. For a labelled graph $G=(V,E,L)$, we
write $v_1\edge{a}{\to} v_2$ iff $L(v_1,v_2)=a$ (for $v_1,v_2\in
V$). We further write $v_1\edge{A}{\to} v_2$ with $A=a_1\cdots a_n$
iff $v_1\edge{a_1}{\to}\cdots\edge{a_n}{\to}v_2$. Then, a path
$\pi = v_1v_2\cdots$ is \Lfair\ if all labels occur infinitely often in this
path, i.e.~for all $a\in L$, for all $k\geq 1,\ \exists k'\ge k$ such that
$L(v_{k'},v_{k'+1})=a$. 
A cycle $\pi$ is called \emph{constant}
if there exists a state $v$ such that
$\pi=v^\omega$. Moreover, a node is a \emph{sink} if its only
outgoing edges are self loops. Then, we can show the following
technical lemma:
\begin{lemma}\label{Lemma:ab=ba}
  Let $G=(V,E,L)$ be a finite complete deterministic labelled graph
  satisfying: for all $v\in V$, for all $a,b\in L$, there are
  $A, B \in L^*$ and $ \tilde{v} \in V$ such that
  $v\edge{aA}{\to} \tilde{v}$ and $v\edge{baB}{\to} \tilde{v}$.  If
  there exists a state from which we can reach two different sinks, then $G$ has a non constant \Lfair\
  cycle.
\end{lemma}

Thanks to this result, we can establish
Theorem~\ref{thm:one-equilibrium}.  We prove the contrapositive, as
follows.  We assume that $\graph{\game}{\BPCdyn}$ has more than one
equilibrium. We introduce a new dynamics $\leadsto$ (taking into
account the \emph{beliefs} of the players about the other players'
strategies) and we use Lemma~\ref{Lemma:ab=ba}, to show that
$\graph\game\leadsto$ has an \Lfair\ cycle.  Then, we define a partial
simulation $\sqsubseteq$ of $\graph\game\leadsto$ by
$\graph{\game}{\BPCdyn}$ and use Proposition~\ref{Cor:weak et
  co-weak=>same termination} to conclude that $\graph{\game}{\BPCdyn}$
has a cycle, which is fair. Hence, $\BPCdyn$ does not fairly
terminate.

\subparagraph{Griffin \textit{et al}: Dispute
  wheels.} Another classical notion in the BGP literature is that of a
\emph{dispute wheel} (DW for short), defined by Griffin \textit{et
  al.}~\cite{Griffin02} as a ``\emph{circular set of conflicting
  rankings between nodes}''. They have shown that the absence of a DW
is a sufficient condition for safety, which is of course of major
practical interest to prove that BGP will converge in a given
network. Moreover, a DW is an instance of a \emph{forbidden pattern}
in a game, and we will thus apply the results from
Section~\ref{Sec:Game}.

We start by formally defining a DW. Let
$\game=(V, E, (V_i)_{1\le i\le n}, (\preceq_i)_{1\le i\le n})$ be a
\OTG with $\perm_i$ the set of permitted paths of $v_i$. A triple
$D =(U, P, H)$ is a DW of $\game$ if:\\
(i) $U=(u_1,\ldots,u_k)\in V^k$ is a tuple of
  states;
(ii) $P=(\pi_1,\ldots,\pi_k)$ is a tuple of permitted paths such
  that for all $ 1\le i \le k$: $\pi_i \in \perm_{u_i}$, i.e., $\pi_i$
  is a permitted path starting in $u_i$ ;
(iii) $H=(h_1,\ldots,h_k)$ is a tuple of non-maximal paths such
  that for all $1\le i \le k$:
  $h_i\pi_{(i\mod k)+1}\in \perm_{u_i}$; and
(iv) for all $1\le i\le k$: $\pi_i\prec_{u_i} h_i \pi_{(i\mod k)+1}$.

Intuitively, in a DW, all players $u_i$ (for $i=1,\ldots, k$) can chose
between two paths to $\target$: either a `direct' path $\pi_i$, or an
`indirect' path $h_i\pi_{(i\mod k)+1}$, which traverses
$u_{(i\mod k)+1}$; and where the latter is always preferred. So $u_1$
prefers to reach through $u_2$, $u_2$ through $u_3$, and so on until
$u_k$ who prefers to reach through $u_1$. Such a conflict clearly
yields loops where the target is never reached. The game in
\figurename~\ref{Fig:DISAGREE+graphDyn} is a typical example of game
that has a DW, if we let $U=(v_1, v_2)$, $P=(v_1\target,v_2\target)$
and \todo{Was $(v_1v_2, v_2v_1)$, but then we have stuttering in the
  path :-) G.}$H=(v_1, v_2)$. Indeed,
$v_1\target\prec_1 v_1v_2\target$ and
$v_2\target\prec_2 v_2v_1\target$. Then, in our setting the sufficient
condition of Griffin \textit{et al.}~\cite{Griffin02} becomes:
\begin{theorem}[\cite{Griffin02}]\label{Thm:Griffin}
  Let $\game$ be a \OTG. If $\game$ has no DW then $\BPCdyn$
  fairly terminates for $\game$.
\end{theorem}

\subparagraph{New result: strong dispute wheels for a necessary
  condition} It is well-known, however, that the absence of a DW is
\emph{not necessary} (see for example
\figurename~\ref{Fig:ce-theorem-9} for a game that has a DW but where
$\BPCdyn$ terminates). As far as we know, finding a \emph{unique and
  necessary} condition for the \emph{fair} termination of $\BPCdyn$ in
\OTG{s} is still an open problem.

Relying on our framework, we manage to obtain such a \emph{necessary
  and sufficient} condition in a restricted setting. We first
strengthen the definition of DW by introducing the notion of
\emph{strong dispute wheel} (SDW for short). We then obtain two
original (as far as we know) results regarding SDW. First, the absence
of SDW is a \emph{necessary} condition for the termination of $\PCdyn$
(i.e.~we drop the best-reply and the fairness hypothesis). Second,
the absence of an SDW is also a \emph{sufficient} condition in the
restricted setting where the preferences of the players range only on
their next-hop. This means for example that $u_1$ prefers to reach the
target through $u_2$ rather than through $u_3$, but does not mind the
route $u_2$ uses (as long as $\target$ is reached). While this is a
restriction, we believe that it is still meaningful in practice, since
networks usually have little control about the routes chosen by their
neighbours.

We first define the notion of SDW. Let $\game$ be a \OTG and
$D =(U, P, H)$ be a DW of $\game$. Then, $D$ is a \emph{strong}
dispute wheel (SDW) of $\game$ if:
\begin{enumerate}
\item for all $1\leq i\leq k$: all states $u_i\in U$ occur \emph{only}
  in $\pi_i$, $h_i$ and $h_{i-1}$ (we identify $h_0$ with $h_k$) and
  not in the other paths of $P$ and $H$; and
\item for all $\pi_i, \pi_j\in P$, for all $h_k, h_\ell\in H$ with
  $k\neq \ell$: $\pi_i$, $h_k$ and $h_\ell$ share no states of
  $V\setminus U$, and if $\pi_i$ and $\pi_j$ share a state $v$ of
  $V\setminus U$ then $\pi_i$ and $\pi_j$ have the same suffix
  after~$v$.
\end{enumerate}
An important property of this definition is that, whenever a game
$\game$ contains an SDW $D=(U,P,H)$, we can extract a minor $\game'$
which is essentially an SDW restricted to the states of $U$ (formally,
$\game'$ contains an SDW $D'=(U',P',H')$ where $U'=U$ is the set of
states of $\game'$). We do so by first deleting from $\game$ all edges
that do not occur in $P$ and $H$; then all $v\not\in U$ (which have at
most one outgoing edge at this point), using the procedure described
in Section~\ref{Sec:Game}. Note that the two extra conditions in the
definition of an SDW guarantee that the deletion of all the states
$v\not\in U$ can occur.

\begin{theorem}\label{thm:sdw->not-termination}\label{Thm:DWAB=>not terminates}
  Let $\game$ be a \OTG. If $\PCdyn$ terminates for $\game$, then
  $\game$ has no SDW.
\end{theorem}
\begin{proof} By Theorem~\ref{Thm:minor of game=>simulation of graph},
  it is sufficient to prove that the dynamics $\PCdyn$ does not
  terminate in the minor game $\game'$ extracted from the SDW (see
  above). We let, for all $1\le i \le k$,
  $\sigma_1(u_i)=u_{(i\mod k)+1}$, and $\sigma_2(u_i)=\target$. Since
  the path resulting from $\sigma_1$ does not visit $v_\bot$, by
  definition of an SDW, we have
  $\sigma_1\PCdyn\sigma_2\PCdyn\sigma_1$. Hence
  $\graph{\game'}{\PCdyn}$ contains a cycle.
\end{proof}
Thus, the absence of an SDW is a \emph{necessary} condition for the
termination of $\PCdyn$. We can further show that this condition is
\emph{sufficient} in the restricted case where any two (permitted)
paths that have the same next-hop are equivalent. Formally, let
$\game$ be a \OTG. We say that it is a \emph{neighbour one target
  game} (\NOTG for short) if for all players $i$, for all permitted
paths $ \pi_1,\pi_2\in \perm_i$ of player $i$: $\pi_1=v_iv\pi_1'$ and
$\pi_2=v_iv\pi_2'$ implies that $\pi_1\eq_i \pi_2$. Then, we can show
the following, relying on Theorem~\ref{Thm:DWAB=>not terminates} (and
thus, also on Theorem~\ref{Thm:minor of game=>simulation of graph}):
\begin{theorem}\label{Thm:Eqiv condition neighbour}
  Let $\game$ be a \NOTG. Then, $\PCdyn$ does not fairly terminate for
  $\game$ if and only if $\game$ has an SDW.
\end{theorem}

\begin{proof}[Sketch of proof]
  In \cite{Griffin02}, Griffin \textit{et al}. prove a stronger result
  than Theorem~\ref{Thm:Griffin}, showing that \emph{if}
  $\graph\game\PCdyn$ has a fair cycle, \emph{then} $\game$ has a DW
  satisfying the following additional properties:
  (1) for all $ u_i\in U$: $j\neq i$ implies $u_i\notin \pi_j$;
  (2) for all $v\notin U$, for all $i,j$:
    $v\notin \pi_i\cap h_j$; and
  (3) for all $v\in \pi_i\cap \pi_j$: $\pi_i(v)=\pi_j(v)$.
  
  We call DW+ such DW. Then, the general schema of our proof is
  summarised in \figurename~\ref{fig:app:dw:summary}: first, we show
  that the existence of a DW+ implies an SDW by showing the required
  additional properties.  By Theorem~\ref{Thm:DWAB=>not terminates},
  this implies $\graph\game\PCdyn$ has a cycle. Then, we conclude by
  showing that this implies the existence of a fair cycle.
\end{proof}

\subparagraph{Finding an SDW in practice} Because of the intricate
definition of SDW, finding an SDW in a real network may be challenging
in practice. However, we have:
\begin{proposition}\label{prop:disagree}
  Let $\game$ be an \NOTG. Then $\PCdyn$ does not fairly terminate for
  $\game$ if and only if $\game^{\mathrm{DIS}}$ is a minor of $\game$.
\end{proposition}

 \begin{figure}[tbp]
 \centering
\scalebox{.8}{\begin{tikzpicture}

\draw (0,1.4) node {$\BPCdyn$ does not fairly};
\draw (0,1.2) node (BFC) [inner xsep=45pt] {\hspace{12mm}};
\draw (0,1) node (BFC1){terminate for $\game$};

\draw (4.5,1.4) node {$\PCdyn$ does not fairly};
\draw (4.5,1.2) node (FC) [inner xsep=45pt] {\hspace{12mm}};
\draw (4.5,1) node (FC1){terminate for $\game$};

\draw (9,1.4) node {$\PCdyn$ does not};
\draw (9,1.2) node (C) [inner xsep=38pt] {\hspace{12mm}};
\draw (9,1) node (C1){terminate for $\game$};

\draw (0,0) node (DW) {$\game$ has a DW};
\draw (4.5,0) node (DWA) {$\game$ has a DW+};
\draw (9,0) node (DWAB) {$\game$ has an SDW};

\draw[->,  thick] (BFC) to (FC);
\draw[->, thick] (FC) to (C);

\draw[->, thick] (DWAB) to (DWA);
\draw[->, thick] (DWA) to (DW);

\draw[->, double,thick] (BFC1) to node[right]{Thm~\ref{Thm:Griffin} \cite{Griffin02}} (DW);
\draw[->, thick] (FC1) to node[left]{\cite{Griffin02}} (DWA);
\draw[->, double, thick] (DWAB) to node[right]{Thm~\ref{thm:sdw->not-termination}} (C1);
\draw[<->,double,thick,dashed](DWAB) to node[above,sloped]{Thm~\ref{Thm:Eqiv condition neighbour}} (FC1);

\end{tikzpicture}}
 \caption{Relationship between SDW and prior
   results: dashed arrow only holds for \NOTG.}
 \label{fig:app:dw:summary}
 \end{figure}
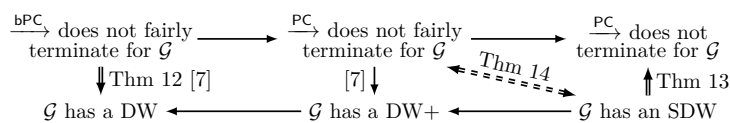

%

\section{Perspectives}\label{Sec:FW}

We envision multiple directions of future work. First, we could
consider games with imperfect information. In the
application to interdomain routing for example, this could be used to
model a malicious router that advertises lies to selected
neighbours. Advertising a non-existent or non-feasible path would allow
for example an attacker to attract the packets of an opponent's
network.
Second, we could investigate a better way to model asynchronicity
(useful for the routing problem) than the concurrent dynamics we have
studied here. Third, we chose to model fairness via a qualitative
property which ensures that \emph{all the players will eventually have
  the opportunity to update their strategies if they want to}. An
alternative way could be the use of probabilities: indeed, there are
games for which a dynamics $\to$ does not fairly terminate, but where
an equilibrium is reached almost surely when interpreting
$\graph\game\to$ as a finite Markov chain (with uniform
distributions).
Finally, we could apply the dynamics of graph-based games to other
problems than interdomain routing, like \emph{load sensitive routing}.


\clearpage

\bibliographystyle{plainurl}

\bibliography{biblio}

 \clearpage
 \appendix

 \appendix

\section{Appendix to Section~\ref{Sec:Game}}\label{sec:append-minor}

\subsection{Missing proofs}

\begin{proof}[Proof of Theorem~\ref{Thm:minor of game=>simulation of graph}] 
Let $\game$ be a game, $\game'$ be a minor of $\game$ and
  ${\to}\in\{\Odyn,\POdyn,\PCdyn\}$. We have to prove that $\graph\game\to$ simulates
  $\graph{\game'}\to$.

Since simulations are preorders, they are transitive
  relations (if $G$ simulates $G'$ and $G'$ simulates $G''$, then
  $G$ simulates $G''$). Therefore, as $\game'$ has been obtain from $\game$ by a (finite) sequence of deletion of edges and nodes, we let $\game=\game_0, \game_1,\ldots,\game_k=\game'$ such that $\forall 1\le i\le k,\ \game_{i+1}$ is a minor of $\game_i$ obtain from $\game_i$ either by deleting a
  single edge, or by deleting a single state. If we prove that, $\forall 1\le i\le k$, $\graph{\game_i}\to$ simulates
  $\graph{\game_{i+1}}\to$ we would have proven our result by transitivity. 
  
  Then, it is sufficient to only consider
  that $\game'$ has been obtained from $\game$ either by deleting a
  single edge, or by deleting a single state.

\vspace*{5mm}

  First, let us consider the case where $\game'$ is obtained by the
  deletion of an edge $e_0$. For the dynamics $\Odyn$, we define a
  relation ${\sqsubseteq}$, subset of
  $\Sigma(\game')\times \Sigma(\game)$, by letting
  $\sigma'\sqsubseteq \sigma$ if for all histories $h\in\hist(\game')$
  (i.e.~$h$ does not go through $e_0$), we have
  $\sigma(h)=\sigma'(h)$. This relation has $\Sigma(\game')$ as a
  domain. We are going to show that this is a simulation relation. To
  do so, consider $\sigma'$ and $\tau'$ in $\Sigma(\game')$ such that
  $\sigma'\Odyn\tau'$, and a strategy $\sigma\in\Sigma(\game)$ such
  that $\sigma'\sqsubseteq \sigma$. Let $\tau\in\Sigma(\game)$ be
  defined for a history~$h$ of $\game$ by $\tau(h) = \tau'(h)$ if
  $h\in\hist(\game')$, and $\tau(h) = \sigma(h)$ otherwise. We must
  prove that $\sigma\Odyn\tau$. Let  $h'\in\hist(\game')$ be the
  history s.t.
  \begin{inparaenum}[(1)]
  \item $\sigma'(h')\neq \tau'(h')$; and
  \item for all histories $h\in\hist(\game')\setminus\{h'\}$:
    $\sigma'(h)=\tau'(h)$.
  \end{inparaenum}
  Such a history $h'$ exists since $\sigma'\Odyn\tau'$.
%
  Moreover, if $i=\player{h'}$, then
  $\Outc{\sigma'}{h'}\prec_i\Outc{\tau'}{h'}$, again because
  $\sigma'\Odyn\tau'$. Then, we have
  $\tau(h') = \tau'(h') \neq \sigma'(h') = \sigma(h')$. Moreover, for
  all $h\in\hist(\game)\setminus\{h'\}$, if $h\notin\hist(\game')$
  then $\tau(h)=\sigma(h)$ by definition, and otherwise
  $\tau(h)=\tau'(h)=\sigma'(h)=\sigma(h)$. Therefore,
  $\Outc{\sigma}{h'}\prec_i\Outc{\tau}{h'}$, so that
  $\sigma\Odyn\tau$.  We thus have proved that
  $\graph{\game'}{\Odyn}\sqsubseteq \graph{\game}{\Odyn}$.
  
  The same kind of reasoning can be made for $\POdyn$ and
  $\PCdyn$. The situation is easier in the case of positional
  strategies, since $\game$ and $\game'$ have the same set of states
  (as we have assumed we have only deleted an edge). Let
  ${\sqsubseteq}$, subset of $\pstrat(\game')\times \pstrat(\game)$ be
  such that $\sigma'\sqsubseteq \sigma$ if
  $\forall v\in V,\ \sigma'(v)=\sigma(v)$ (note that, although
  $\sigma$ and $\sigma'$ take the same decisions in the same states,
  they are, formally, different objects. Indeed, they map histories,
  which are different in $\game$ and $\game'$, to states). This
  relation has $\pstrat(\game')$ as domain, and for all
  $\sigma'\in \pstrat(\game')$ there is one and only one
  $\sigma\in \pstrat(\game)$ such that $\sigma'\sqsubseteq \sigma$.
  We are going to show that this is a simulation relation. To do so,
  consider $\sigma'$ and $\tau'$ in $\pstrat(\game')$ such that
  $\sigma'\POdyn\tau'$ (resp. $\sigma'\PCdyn\tau'$), and a strategy
  $\sigma\in\pstrat(\game)$ such that $\sigma'\sqsubseteq \sigma$. Let
  $\tau\in\pstrat(\game)$ be the (unique) strategy profile such that
  $\tau'\sqsubseteq \tau$. By definition of $\POdyn$ (resp., $\PCdyn$)
  and $\sqsubseteq$, we have $\sigma\POdyn \tau$ (resp.,
  $\sigma\PCdyn \tau$).
  \vspace*{5mm}
  
  Second, consider the case where $\game'$ is obtained by the deletion of
  a state $v$. If $v$ is an isolated node, the result is trivially true as $\graph\game\to =\graph{\game'}\to$. Then, we suppose that $v$ is not an isolated node. By definition of minor, it means that $v$ has a unique outgoing edge $(v,v')$ and all
    predecessors $u$ of $v$ do not have $v'$ as a
    successor.
  
   For the dynamics $\Odyn$, if $h\in\hist(\game)$ such
  that $last(h)\neq v$, we can construct a corresponding play $f(h)$
  of $\game'$ by replacing a sequence $uvv'$ of $h$ by $uv'$. Notice
  that $f$ is a bijection from
  $\hist(\game)\setminus \{h\mid last(h)=v\}$ to $\hist(\game')$, by
  the conditions over the deletion of $v$ in the definition of minors.
  Then, the relation on strategy profiles is given by
  $\sigma'\sqsubseteq \sigma$ if:
  \begin{align*}
    \text{for all }h\in \hist(\game)\setminus \{h\mid last(h)=v\} :
    \sigma'(f(h))
    &=
      \begin{cases}
        \sigma(h) &\text{if }\sigma(h)\neq v\\
        \sigma'(f(h))=v'&\text{otherwise.}
      \end{cases}        
  \end{align*}
  Since $f$ is bijective and since, for all $\sigma$ and $h$:
  $last(h)=v$ implies $\sigma(h)=v'$; then the relation $\sqsubseteq$
  is also a bijection. This enables us to prove as before that it is a
  simulation relation (indeed even a bisimulation relation), so that
  $\sqsubseteq$ is a simulation relation of $\graph{\game'}{\Odyn}$ by
  $\graph{\game}{\Odyn}$, in that case too.
  
  We can make the same reasoning for $\POdyn$ and $\PCdyn$. Let
  $\sqsubseteq$, subset of $\pstrat(\game')\times \pstrat(\game)$
  such that $\sigma'\sqsubseteq \sigma$ if:
  \begin{align*}
    \text{for all } u\in V\setminus \{v\} : \sigma'(u)
    &=
      \begin{cases}
        \sigma(u) &\text{if }\sigma(u)\neq v\\
        v' &\text{if }\sigma(u)=v.
      \end{cases}
  \end{align*}
  It is clear that this relation is a bisimulation between
  $\graph{\game'}\POdyn$ and $\graph\game\POdyn$
  (resp. $\graph{\game'}\PCdyn$ and $\graph\game\PCdyn$).
\end{proof}


\begin{proof}[Proof of Theorem \ref{Thm:dominant minor}]

Let $\game$ be a game, $\game'$ be a dominant minor of $\game$ and 
  ${\to} \in \{\BPOdyn,\BPCdyn\}$. We have to prove that we can build a simulation
  $\sqsubseteq$ of $\graph{\game'}{\to}$ by $\graph{\game}{\to}$ such
  that:
  \begin{inparaenum}[(i)]
  \item $\sqsubseteq^{-1}$ is a partial simulation of
    $\graph{\game}{\to}$ by $\graph{\game'}{\to}$; and
  \item if there is a fair
  cycle in $\game$ then there is a fair cycle in $\game'$.
  \end{inparaenum}
  In particular, the dynamics $\to$ fairly terminates for $\game$ if
  and only if it does for $\game'$.

  We follow the same proof as in Theorem~\ref{Thm:minor of
  game=>simulation of graph}. We have already seen that if $\game'$
  is obtained by the deletion of a state, then $\game$ and $\game'$
  are bisimilar, and we can show easily that a fair cycle in
  $\graph\game{\to}$ implies the existence of a fair cycle in
  $\graph{\game'}{\to}$.

  If $\game'$ is obtained by the deletion of a \emph{dominated} edge
  $e_0$, then we construct the same simulation relation $\sqsubseteq$ as
  before (in the cases $\POdyn$ and $\PCdyn$),
  i.e. $\sigma'\sqsubseteq \sigma$ if $\forall v\in V$:
  $\sigma'(v)=\sigma(v)$.  This is still a simulation relation, even
  with the best-reply dynamics, since we remove a dominated edge. We
  can also show similarly that $\sqsubseteq^{-1}$ is a partial
  simulation, by using once again that a change of profile towards $e_0$
  is not possible since it is dominated. \todo{M:Je ne pense pas que ce
    soit nécessaire d'être plus formel...} If there is a fair cycle in
  $\graph{\game'}{\to}$ then the simulation gives a cycle in
  $\graph{\game}{\to}$: this is also a fair cycle since the edge $e_0$
  is dominated and thus cannot be the unique cause of
  non-fairness. Reciprocally, the partial simulation
  $\sqsubseteq^{-1}$ allows one to reconstruct from a fair cycle in
  $\graph{\game}{\to}$ a cycle in $\graph{\game'}{\to}$: if this cycle
  is not fair, this means that the edge $e_0$ is chosen infinitely often
  in the cycle of $\game$, which violates the best-reply condition of
  the dynamics $\to$ in $\game$.
\end{proof}

 \subsection{Examples related to Theorem \ref{Thm:minor of
     game=>simulation of graph} and
   Theorem~\ref{theo:fair-termination}}

As seen in example \figurename~\ref{Fig:ce-theorem-9}, Theorem~\ref{Thm:minor of game=>simulation of graph}, does not works for $\BPCdyn$ and $\BPOdyn$.

%
%

\begin{figure}\centering
\begin{tikzpicture}
	\draw (0,2) node [state] (j1) {$v_1$};
	\draw (3,2) node [state] (j2) {$v_2$};
	\draw (1.5,0) node [state] (cible) {$v_\bot$};
        \draw (1.5,1.5) node[state, very thick] (j3) {$v_3$} ;

	\draw[->,bend left=10]  (j1) to node[above]{$c_1$} (j2);
	\draw[->,bend left=10, very thick]  (j2) to node[below]{$c_2$}
        (j3);
	\draw[->,bend left=10, very thick]  (j3) to node[below]{$c_3$} (j1);	
	\draw[->]  (j1) to node[left] {$s_1$} (cible);
	\draw[->]  (j2) to node[right] {$s_2$} (cible);
        \draw[->, very thick]  (j3) to node[right] {$s_3$} (cible);

\end{tikzpicture}\hspace*{10mm}
\begin{tikzpicture}
  \draw[fill=gray!20] (-1, -.5) rectangle (3,2.5) ;

	\draw (0,2) node [state,rectangle] (cc) {$c_1c_2 c_3$};
	\draw (2,2) node [state,rectangle] (sc) {$s_1c_2 c_3$};
	\draw (0,0) node [state,rectangle] (cs) {$c_1s_2 c_3$};
	\draw (2,0) node [state,rectangle] (ss) {$s_1s_2 c_3$};

    \draw[->] (cc) to (sc);
    \draw[->] (cc) to (cs);
    \draw[->] (ss) to (sc);    
    \draw[->] (ss) to (cs);
    
    \draw[->, bend left=10] (ss) to (cc);    
    \draw[->, bend left=10] (cc) to (ss);

    \draw[fill=gray!20] (4, -.5) rectangle (8,2.5) ;

	\draw (5,2) node [state,rectangle] (ccs) {$c_1c_2 s_3$};
	\draw (7,2) node [state,rectangle] (scs) {$s_1c_2 s_3$};
	\draw (5,0) node [state,rectangle] (css) {$c_1s_2 s_3$};
	\draw (7,0) node [state,rectangle] (sss) {$s_1s_2 s_3$};

    \draw[->] (ccs) to (scs);
    \draw[->] (ccs) to (css);

    \draw[->, bend left=10] (ccs) to (sss);

    \draw[->, bend left=10] (3, 2) to (4,2) ;
    \draw[->, bend right=10] (3, 0) to (4,0) ;
  \end{tikzpicture}

  \caption{Left: a 3-player game $\game$ with $\game^{\mathrm{DIS}}$
    (\figurename~\ref{Fig:DISAGREE+graphDyn}) as a minor. Right:
    $\graph{\game}\PCdyn$ with a non-fair cycle.}
   \label{Fig:ce-theorem-10-1}
 \end{figure}

 \figurename~\ref{Fig:ce-theorem-10-1} shows why
 Theorem~\ref{Thm:minor of game=>simulation of graph} does not work
 with fair termination (instead of termination).
 The game $\game$ on the left of the figure is a $3$-player game,
 where Player $i$ owns state $v_i$ (for $i=1,2,3$), and where the
 players' preferences are as follows:
 \begin{itemize}
 \item $v_1v_\bot\prec_1v_1v_2v_\bot$;
 \item $v_2v_\bot\prec_2v_2v_3v_1v_\bot$; and
 \item Player 3 prefers $v_3v_\bot$ to all other paths.
 \end{itemize}
 Clearly, $\game^{\mathrm{DIS}}$ (\figurename~\ref{Fig:DISAGREE+graphDyn}) is a
 minor of $\game$. We already know
 (\figurename~\ref{Fig:DISAGREE+graphDyn} again) that
 $\graph{\game^{\mathrm{DIS}}}{\PCdyn}$ admits a \emph{fair} cycle, so, if
 Theorem~\ref{theo:fair-termination} were to hold on $\PCdyn$,
 $\graph{\game}{\PCdyn}$ should admit a \emph{fair} cycle too. Let us
 explain why it is not the case.  The graph $\graph{\game}{\PCdyn}$ is
 sketched on the right-hand side of
 \figurename~\ref{Fig:ce-theorem-10-1}. All the nodes in the left gray
 boxes are those where Player 3 plays $c_3$, and those in the right
 box are where he plays $s_3$. Several edges are not drawn for
 clarity: they are all the edges where Player $3$ changes from $c_3$
 to $s_3$. Given his preferences, he will \emph{always} choose to do
 so if he decides to update his strategy, so there is an implicit edge
 from $c_1c_2c_3$ to $c_1c_2s_3$, but also one from $c_1c_2c_3$ to
 $c_1s_2s_3$ for example. Then, as soon as Player $3$ has decided to
 play $s_3$, the other two players will both prefer to play $s_1$ and
 $s_2$ respectively. As a consequence, the only cycle is 
 $(c_1c_2c_3,s_1s_2c_3)^\omega$, where Player $3$ never plays although he
 could. This cycle is not fair. Hence, $\graph{\game}{\PCdyn}$ admits
 no fair cycle.

 \begin{figure}
   \centering
   \begin{tikzpicture}
     \node[state] at (0,0) (v4) {$v_4$} ;
     \node[state] at (1.5,0) (t) {$v_\bot$} ;
     \node[state] at (1.5,1.5) (v1) {$v_1$} ;
     \node[state] at (3, 0) (v2) {$v_2$} ;
     \node[state] at (3,1.5) (v3) {$v_3$} ;

     \path[->]
     (v1) edge[very thick] node[left] {$e$} (t)
          edge node[above] {$f$} (v4)
          edge (v2)
     (v2) edge (v3)
          edge (t)
     (v3) edge (v1)
          edge (t)
     (v4) edge (t) ;
   \end{tikzpicture}
   \caption{A game $\game$ where $\graph\game\PCdyn$ does not
     terminate. When removing the edge $e$ (which is dominated by
     $f$), we obtain a new game $\game'$ which is a dominant minor
     of this game, in which $\graph{\game'}{\PCdyn}$
     terminates. \label{fig:ce-theorem-10-2}}
 \end{figure}

\vspace*{1mm}

 Finally, we exhibit a third example of a game $\game$ in which
 $\graph\game\PCdyn$ does not terminate, and from which we can extract
 a dominant minor $\game'$ s.t. $\graph{\game'}{\PCdyn}$
 terminates. This shows why Theorem~\ref{Thm:dominant minor} does not
 hold anymore when the `best reply' hypothesis is dropped. This game
 is showed in \figurename~\ref{fig:ce-theorem-10-2}, with the
 following preferences for the players:
 \begin{itemize}
 \item $v_1v_2v_3v_\bot\prec_1 v_1v_\bot\prec_1v_1v_4v_\bot\prec_1
   v_1v_2v_\bot$;
 \item
   $v_2v_3v_1v_4v_\bot\prec_2 v_2v_\bot\prec_2 v_2v_3v_1v_\bot\prec_2
   v_2v_3v_\bot$;
 \item
   $v_3v_1v_4v_\bot\prec_3 v_3v_1v_2v_\bot\prec_3 v_3v_\bot\prec_3
   v_3v_1v_\bot$.
 \end{itemize}
 
 Then, one can check that the following sequence of (positional)
 strategy profiles forms a cycle for $\PCdyn$ (but not for
 $\BPCdyn$). Note that we do not indicate $\sigma_4$ which is always
 equal to $v_\bot$, and we write in bold the change from the previous
 profile:
 \begin{enumerate}
 \item $\sigma_1=v_2$, $\sigma_2=v_\bot$, and $\sigma_3=v_\bot$;
 \item $\sigma_1=v_2$, $\pmb{\sigma_2=v_3}$, and $\sigma_3=v_\bot$;
 \item $\pmb{\sigma_1=v_\bot}$, $\sigma_2=v_3$, and $\sigma_3=v_\bot$;
 \item $\sigma_1=v_\bot$, $\sigma_2=v_3$, and $\pmb{\sigma_3=v_1}$;
 \item $\pmb{\sigma_1=v_4}$, $\sigma_2=v_3$, and $\sigma_3=v_1$;
 \item $\sigma_1=v_4$, $\pmb{\sigma_2=v_\bot}$, and $\sigma_3=v_1$;
 \item $\sigma_1=v_4$, $\sigma_2=v_\bot$, and $\pmb{\sigma_3=v_\bot}$;
 \item $\pmb{\sigma_1=v_2}$, $\sigma_2=v_\bot$, and $\sigma_3=v_\bot$;
 \end{enumerate}

 However, when removing from $\game$ the edge $e$ (which is dominated
 by $f$), one obtains a new game $\game'$ that is a dominant minor of
 $\game$ and in which the $\PCdyn$ terminates.
 
 

\section{Appendix to
  Section~\ref{sec:applications}} \label{app:applications}
\subsection{Termination implies a unique equilibrium}

\begin{proof}[Proof of Lemma~\ref{Lemma:ab=ba}]

Let $G=(V,E,L)$ be a finite complete deterministic labelled graph
  satisfying: for all $v\in V$, for all $a,b\in L$, there are
  $A, B \in L^*$ and $ \tilde{v} \in V$ such that
  $v\overset{aA}{\to} \tilde{v}$ and $v\overset{baB}{\to} \tilde{v}$, and such that there exists a state from which we can reach two different sinks. We have to prove that $G$ has a non constant \Lfair\ cycle.

  Let $\{t_1,\ldots t_k\}$ be the set of sinks. Given $t_i$
  a sink, we define $T_i$ as the set of states from which
  all paths lead to $t_i$, i.e.
  \begin{align*}
    T_i&=\{v\in V\mid \forall w\in V, (v,w)\in E,\text{ implies that }
         w\in T_i\text{ and there is a path from } v \text{ to } t_i\}.
  \end{align*}
  Let $T_0=V\setminus(T_1\cup\ldots\cup T_k)$. By hypothesis,
  $T_0\neq \emptyset$.

Let $L=(l_1,\ldots,l_k)$ and let $v_0\in T_0$. From $v_0$, we define $\pi=v_0,v_1,\ldots$ such that $\forall i\ge 0,\ v_i\overset{l_{(i+1)\mod k}}{\longrightarrow} v_{i+1}$. As $G$ is finite, $\pi$ is a lasso, ie $\exists \pi_1,\ \pi_2$ finite paths such that $\pi=\pi_1\pi_2^\omega$. In particular, let $\pi'=\pi_2^\omega$ be a cycle. Moreover, $\pi$ is \Lfair\ by construction, then so is $\pi'$. Two possibilities : either $\pi'$ is not constant, and the Lemma is then satisfies as we have found a non constant \Lfair\ cycle in $G$; or $\pi'$ is contant, which means that $\pi_2=t_i$ for some $i$. In this case, $v_0$ has a path to some $t_i$.

To conclude the proof, we just have to consider the case when every $v\in T_0$ has a path to some $t_i$.

Moreover, for $v\in T_0$, if $v$ cannot reach any other $t_j$, it means that $v\in T_i$, by definition of $T_i$. Then we can suppose that from every $v\in T_0$, we
  can reach at least two terminal states.

  In order to prove that there is a non constant \Lfair\ cycle in $G$,
   we just have to prove that, from any $v\in T_0$, there exists a
  \Lfair\ path that stays in $T_0$, i.e.
  $\forall v\in T_0$, $\exists \mathcal{A}$ such that if
  $\forall a\in L$: $a\in \mathcal{A}$ and
  $v\overset{\mathcal{A}}{\to} v'\in T_0$. In fact, it is sufficient
  to prove that
  $\forall v\in T_0$, $\forall a\in L$, $\exists \mathcal{A}$ such that
  $a\in \mathcal{A}$ and such that
  $v\overset{\mathcal{A}}{\to} v'\in T_0$. The proof of this equivalence is made by iteration and is left to the reader.

  Towards a contradiction, let us assume that $\exists v\in T_0$,
  $\exists a\in L$ s.t. $\forall \mathcal{A}$: if $a\in \mathcal{A}$,
  then $v\overset{\mathcal{A}}{\to} v'\notin T_0$. In particular,
  $v\overset{a}{\to}v_1\in T_i$ with $i\neq 0$. Notice that this path
  exists because the graph is complete, and is unique because the
  graph is deterministic. As $v\in T_0$, it means that exists a path
  to $t_j$ with $j\neq 0$ and $j\neq i$. Let
  $v\to v_1\to\ldots\to v_k$ with $v_{k-1}\in T_0$ and $v_k\in
  T_j$. $\forall l\in \{1,\ldots,k-1\}$, let $s_l$ such that
  $v_l\overset{a}{\to}s_l$. By hypothesis, $s_l\notin T_0$ (otherwise there exists a path from $v$ and containing $a$ that leads to some $v'\in T_0$). Then let
  $\ell$ be such that $s_\ell\in T_{i'}$ and $s_{\ell+1}\in T_{j'}$
  such that $i'\neq j'$; and let us assume that
  $v_\ell\overset{b}{\to} v_{\ell+1}$. By hypothesis, there are
  $\mathcal{A}$, $\mathcal{B}$ and $v^*$ such that:
  \begin{itemize}
  \item $v_\ell\overset{a}{\to} s_\ell \overset{\mathcal{A}}{\to}v^*$
    with $v^*\in T_{i'}$ (the latter holds since $s_\ell\in T_{i'}$); and
  \item
    $v_\ell\overset{b}{\to} v_{\ell+1}\overset{a}{\to}
    s_{\ell+1}\overset{\mathcal{B}}{\to} v^*$ with $v^*\in T_{j'}$
    (the latter holds since $s_{\ell+1} \in T_{j'}$).
  \end{itemize}
  This leads to a contradiction as $T_{i'}\cap T_{j'}=\emptyset$. This
  situation is modelled by \figurename~\ref{Fig:Thm 2 sink}

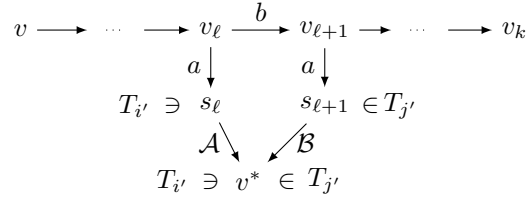
\begin{figure}
  \begin{center}
    \begin{tikzpicture}
      \draw (1,2) node (p0){$v$};
      \draw (2,2) node (r1){$ $};
      \draw (2.5,2) node (r2){$ $};
      \draw (3.5,2) node (pi1){$v_\ell$};
      \draw (5,2) node (pi){$v_{\ell+1}$};
      \draw (6,2) node (r3){};
      \draw (6.5,2) node (r4){};
      \draw (7.5,2) node (pk){$v_k$};
      
      \draw (3.5,1) node (qi1){$s_\ell$};
      \draw (5,1) node (qi){$s_{\ell+1}$};

      \draw (3,1) node (in){$\ni$};
      \draw (5.6,1) node (in){$\in$};
      
      \draw (2.5,1) node (red){$T_{i'}$};
      \draw (6,1) node (blue){$T_{j'}$};
      
      \draw[->] (p0) to (r1);
      \draw[dotted] (r1) to (r2);
      \draw[->] (r2) to (pi1);
      \draw[->] (pi1) to node[above]{$b$} (pi);
      \draw[->] (pi) to (r3);
      \draw[dotted] (r3) to (r4);
      \draw[->] (r4) to (pk);
      
      \draw[->] (pi1) to node[left]{$a$} (qi1);
      \draw[->] (pi) to node[left]{$a$} (qi);
      
      \draw (4, 0) node (sigma){$v^*$};
      \draw [->] (qi1) to node[left]{$\mathcal{A}$}(sigma);
      \draw [->](qi) to node[right]{$\mathcal{B}$}(sigma);
      
      \draw (3.5,0) node (in){$\ni$};
      \draw (4.5,0) node (in){$\in$};
      
      \draw (3,0) node (red){$T_{i'}$};
      \draw (5,0) node (blue){$T_{j'}$};
    \end{tikzpicture}
  \end{center}
   \caption{Construction of the contradiction in the case of Lemma\ref{Lemma:ab=ba} \label{Fig:Thm 2 sink}}
 \end{figure}
\end{proof}

\begin{proof}[Proof of Theorem~\ref{thm:one-equilibrium}]

Let $\game$ be a \OTG\ such that ${\BPCdyn}$ fairly terminates for $\game$
  (i.e.~the corresponding instance of SPP is safe). We have to prove that $\game$ has
  exactly one equilibrium.
  
  We prove the contrapositive of the theorem, i.e. we assume that
  $\game$ has more than one equilibrium, and we deduce that
  $\graph{\game}{\BPCdyn}$ has a fair cycle.

  To prove that $\graph{\game}{\BPCdyn}$ has a fair cycle, we proceed
  in two steps. First, we introduce a new dynamics $\leadsto$ and
  prove, by means of Lemma~\ref{Lemma:ab=ba}, that
  $\graph\game\leadsto$ has a non constant \Lfair\ cycle. Then, we define a
  partial simulation $\sqsubseteq$ of $\graph{\game}{\BPCdyn}$ by
  $\graph\game\leadsto$ and use Proposition~\ref{Cor:weak et
    co-weak=>same termination} to conclude that
  $\graph{\game}{\BPCdyn}$ has a cycle. Finally, we show that this
  cycle is fair.

  The new dynamics we consider takes explicitly into account the fact
  that, when a player updates its strategy, it might take some time
  for the other players to be notified. Hence, each state in this new
  dynamics stores what each Player $j$ believes about the current
  strategy of all Player $i$. Formally, let
  $\graph\game\leadsto =(V,\leadsto,L)$ be the labelled graph defined
  as follows.
  \begin{itemize}
  \item the set of states $V$ contains all states $v$ of the form
    $v=(\sigma_{i,j})_{1\le i,j\le n}$, i.e. a state associated a
    strategy $\sigma_{i,j}$ to each \emph{pair} of players $i$ and
    $j$. Intuitively, $\sigma_{i,j}$ represents what Player $j$
    \emph{believes} about the current strategy of Player $i$. In other
    words, for all $1\le j\le n$,
    $\sigma_j=(\sigma_{i,j})_{1\le i\le n}$ represents the belief of
    the player $j$ about the strategies of the other players. In
    particular, $\sigma_{j,j}$ is the current strategy of player
    $j$. We let
    \begin{align*}
      V_0&=\{v\in V\mid \forall 1\le j,j'\le n: (\sigma_{i,j})_{1\le i\le n}=(\sigma_{i,j'})_{1\le i\le n}\}
    \end{align*}
    be the set of states where every player knows the strategies of
    the other players;

  \item $L=\{0, 1,\ldots, n\}$, where $1,\ldots,n$ represents the players; and
  \item $v\overset{\ell}{\leadsto}v'$ iff:
    \begin{itemize}
    \item Either $\ell=0$, then $v'\in V_0$ and for all
      $1\le j\le n,$: $\sigma_{j,j}=\sigma'_{j,j}$. The $0$ change is
      an update of knowledge. All players learn the strategies of the
      other players. Notice that if $v\in V_0$, then $v\overset{0}{\leadsto}v$.
    \item Or $\ell\neq 0$, then for all $j\neq \ell$:
      $(\sigma_{i,j})_{1\le i\le n}=(\sigma'_{i,j})_{1\le i\le n}$ and
      the update of Player $\ell$ is performed according to $\POdyn$
      if possible. Formally,
      $(\sigma_{i,\ell})_{1\le i\le n}\POdyn(\sigma'_{i,\ell})_{1\le
        i\le n}$ with $\sigma_{\ell,\ell}\neq\sigma_{\ell,\ell}'$ if this is
      permitted by $\POdyn$, or
      $(\sigma_{i,\ell})_{1\le i\le n}=(\sigma'_{i,\ell})_{1\le i\le n}$
      otherwise. Such a move in the dynamics thus corresponds to
      Player $\ell$ updating his strategy according to $\POdyn$ if
      possible, or no update at all if not. This change is not yet
      learned by the other players.
  \end{itemize}
  \end{itemize}
  
  \begin{example}

  \begin{figure}\centering
\begin{tikzpicture}
	\draw (0,7.5) node [circle, double, draw, inner sep=0pt] (aa) {$
\begin{matrix}
\mathbf{c_1}&c_2\\
c_1&\mathbf{c_2}
\end{matrix}$};
	\draw (2.5,7.5) node [circle, draw, inner sep=0pt] (ab) {$\begin{matrix}
\mathbf{s_1}&c_2\\
c_1&\mathbf{c_2}
\end{matrix}$};
	\draw (5,7.5) node [circle, draw, inner sep=0pt] (ac) {$
\begin{matrix}
\mathbf{c_1}&c_2\\
c_1&\mathbf{s_2}
\end{matrix}$};
	\draw (7.5,7.5) node [circle, draw, inner sep=0pt] (ad) {$
\begin{matrix}
\mathbf{s_1}&c_2\\
c_1&\mathbf{s_2}
\end{matrix}$};

	\draw (0,5) node [circle, draw, inner sep=0pt] (ba) {$\begin{matrix}
\mathbf{c_1}&c_2\\
s_1&\mathbf{c_2}
\end{matrix}$};
	\draw (2.5,5) node [circle, double, draw, inner sep=0pt] (bb) {$
\begin{matrix}
\mathbf{s_1}&c_2\\
s_1&\mathbf{c_2}
\end{matrix}$};
	\draw (5,5) node [circle, draw, inner sep=0pt] (bc) {$\begin{matrix}
\mathbf{c_1}&c_2\\
s_1&\mathbf{s_2}
\end{matrix}$};
	\draw (7.5,5) node [circle, draw, inner sep=0pt] (bd) {$
	\begin{matrix}
\mathbf{s_1}&c_2\\
s_1&\mathbf{s_2}
\end{matrix}$};

	\draw (0,2.5) node [circle, draw, inner sep=0pt] (ca) {$\begin{matrix}
\mathbf{c_1}&s_2\\
c_1&\mathbf{c_2}
\end{matrix}$};
	\draw (2.5,2.5) node [circle, draw, inner sep=0pt] (cb) {$
\begin{matrix}
\mathbf{s_1}&s_2\\
c_1&\mathbf{c_2}
\end{matrix}$};
	\draw (5,2.5) node [circle, double, draw, inner sep=0pt] (cc) {$\begin{matrix}
\mathbf{c_1}&s_2\\
c_1&\mathbf{s_2}
\end{matrix}$};
	\draw (7.5,2.5) node [circle, draw, inner sep=0pt] (cd) {$
	\begin{matrix}
\mathbf{s_1}&s_2\\
c_1&\mathbf{s_2}
\end{matrix}$};

	\draw (0,0) node [circle, draw, inner sep=0pt] (da) {$\begin{matrix}
\mathbf{c_1}&s_2\\
s_1&\mathbf{c_2}
\end{matrix}$};
	\draw (2.5,0) node [circle, draw, inner sep=0pt] (db) {$
\begin{matrix}
\mathbf{s_1}&s_2\\
s_1&\mathbf{c_2}
\end{matrix}$};
	\draw (5,0) node [circle, draw, inner sep=0pt] (dc) {$\begin{matrix}
\mathbf{c_1}&s_2\\
s_1&\mathbf{s_2}
\end{matrix}$};
	\draw (7.5,0) node [circle, double, draw, inner sep=0pt] (dd) {$
	\begin{matrix}
\mathbf{s_1}&s_2\\
s_1&\mathbf{s_2}
\end{matrix}$};

	\draw[->] (aa) to node[above]{1}(ab);	
	\draw[->, bend left] (aa) to node[above]{2} (ac);
	\draw[->, loop left] (aa) to node[left]{0} (aa);
\end{tikzpicture}
\caption{Representation of a part of $\graph{\game^{DIS}}{\leadsto}$}
   \label{Fig:G^DISleadsto}
\end{figure}

  See \figurename~\ref{Fig:G^DISleadsto} for a representation of a part of $\graph{\game^{DIS}}\leadsto$. Every nodes of $\graph{\game^{DIS}}\leadsto$ are represented, but not every edges, for seeks of clarity.
  
   For example, the node \begin{tikzpicture}
	\draw (2.5,7.5) node [circle, draw, inner sep=0pt] (b) {$\begin{matrix}
\mathbf{s_1}&c_2\\
c_1&\mathbf{c_2}
\end{matrix}$}; 
\end{tikzpicture} represents the state where player $1$ plays $\mathbf{s_1}$ and believes that player $2$ plays $c_2$ (first line) and where player $2$ plays $\mathbf{c_2}$ and believes that player $1$ plays $c_1$ (second line).

The nodes of $V_0$ are represented by the double circle, where we can see that all the believes correspond to the real choice of the players.

Notice that, by definition, the graph is complete deterministic. 

\begin{minipage}{0.69\textwidth}As an example, we considere the three updates possible from \end{minipage}\begin{minipage}{0.2\textwidth}
 \begin{tikzpicture}
	\draw (2.5,7.5) node [circle, double, draw, inner sep=0pt] (b) {$\begin{matrix}
\mathbf{c_1}&c_2\\
c_1&\mathbf{c_2}
\end{matrix}$}; 
\end{tikzpicture}  
\end{minipage}

\begin{minipage}{0.3\textwidth}
\begin{tikzpicture}
	\draw (0,7.5) node [circle, double, draw, inner sep=0pt] (a) {$\begin{matrix}
\mathbf{c_1}&c_2\\
c_1&\mathbf{c_2}
\end{matrix}$}; 

	\draw (2.5,7.5) node [circle, double, draw, inner sep=0pt] (b) {$\begin{matrix}
\mathbf{c_1}&c_2\\
c_1&\mathbf{c_2}
\end{matrix}$};

\draw[->] (a) to node[above]{0}(b);	
\end{tikzpicture}\end{minipage}\begin{minipage}{0.69\textwidth} means that an update of knowledge will stay in the same state.\end{minipage}

\begin{minipage}{0.3\textwidth}
\begin{tikzpicture}
	\draw (0,7.5) node [circle, double, draw, inner sep=0pt] (a) {$\begin{matrix}
\mathbf{c_1}&c_2\\
c_1&\mathbf{c_2}
\end{matrix}$}; 

	\draw (2.5,7.5) node [circle, draw, inner sep=0pt] (b) {$\begin{matrix}
\mathbf{s_1}&c_2\\
c_1&\mathbf{c_2}
\end{matrix}$};

\draw[->] (a) to node[above]{1}(b);	
\end{tikzpicture} \end{minipage}\begin{minipage}{0.69\textwidth} means that player $1$ updates his strategy according to $\POdyn$, changing $c_1$ by $s_1$, but player $2$ has not been informed of this change yet. \end{minipage}

\begin{minipage}{0.3\textwidth}

\begin{tikzpicture}
	\draw (0,7.5) node [circle, double, draw, inner sep=0pt] (a) {$\begin{matrix}
\mathbf{c_1}&c_2\\
c_1&\mathbf{c_2}
\end{matrix}$}; 

	\draw (2.5,7.5) node [circle, draw, inner sep=0pt] (b) {$\begin{matrix}
\mathbf{c_1}&c_2\\
c_1&\mathbf{s_2}
\end{matrix}$};

\draw[->] (a) to node[above]{2}(b);	
\end{tikzpicture} 
\end{minipage}\begin{minipage}{0.69\textwidth} means that player $2$ updates his strategy according to $\POdyn$, changing $c_2$ by $s_1$, but player $1$ has not been informed of this change yet.

\end{minipage} 

Moreover, notice that the sinks are \begin{tikzpicture}
	\draw (2.5,7.5) node [circle, double, draw, inner sep=0pt] (b) {$\begin{matrix}
\mathbf{s_1}&c_2\\
s_1&\mathbf{c_2}
\end{matrix}$}; 
\end{tikzpicture} and \begin{tikzpicture}
	\draw (2.5,7.5) node [circle, double, draw, inner sep=0pt] (b) {$\begin{matrix}
\mathbf{c_1}&s_2\\
c_1&\mathbf{s_2}
\end{matrix}$}; 
\end{tikzpicture}, which correspond in a certain sense to the equilibria in $\graph{\game^{DIS}}\PCdyn$.
  \end{example}


  The rest of the proof will be as follows:
  \begin{enumerate}
  \item We check that $\graph\game\leadsto$ verifies the conditions in
    Lemma~\ref{Lemma:ab=ba}, and then has a \Lfair\ cycle;
  \item We define a partial simulation $\sqsubseteq$ of
    $\graph\game{\leadsto^+}$ by $\graph\game{\BPCdyn^+}$ of domain
    $V_0$;
  \item We deduce that $\graph\game{\leadsto^+}$ has a cycle
    containing elements of $V_0$;
  \item Using Proposition~\ref{Cor:weak et co-weak=>same termination},
    we deduce that $\graph\game{\BPCdyn^+}$ has a cycle;
  \item Then, we conclude that $\graph\game\BPCdyn$ has a fair cycle.
  \end{enumerate}

  \begin{enumerate}
  \item \begin{itemize}
    \item $\graph\game\leadsto$ is deterministic because the dynamics is
      Best Reply and the preferences are strict for a different neighbour.
    \item $\graph\game\leadsto$ is complete by definition.
    \item For $\sigma$ a terminal state of $\game$, let $v\in V_0$ such that $v=(\sigma{i,j})_{1\le i,j\le n}$ with $\sigma_{i,i}=\sigma_i$. Then $v$ is a sink of $\graph\game\leadsto$, as $v\overset{0}{\leadsto}$ and $\sigma$ is terminal for $\POdyn$. Then, as $\graph\game\POdyn$ as several terminal nodes, $\graph\game\leadsto$ has several sinks. 
    \item Let $v\in V$ and let $i,j\in L$. We will show that $\exists
      \mathcal{I},\ \mathcal{J}$ and $v^*$ such that
      $v\overset{i\mathcal{I}}{\leadsto} v^*$ and $v\overset{ji\mathcal{J}}{\leadsto}
      v^*$ \begin{itemize}
      \item If $i=j$, let $\mathcal{I}=\mathcal{J}=\varepsilon$. We have to prove that $\overset{i}{\leadsto} =\overset{ii}{\leadsto}$. Two cases are possible for $i$. Either $i=0$, which means a knowledges update, or $i>0$, which means a player update. In the first case, if $i=0$, $v\overset{0}{\leadsto}v^*$ means that $v^*\in V_0$, by definition of $\overset{0}{\leadsto}$. And as $v^*\in V_0,\ v^*\overset{0}{\leadsto}v^*$, then $v\overset{0}{\leadsto} v^*\overset{0}{\leadsto} v^*$, and then $\overset{0}{\leadsto} =\overset{0}{\leadsto}\overset{0}{\leadsto}$. In the second case, when $i>0,\ \overset{i}{\leadsto}$ represent the update of player $i$. As there is no other change between the two $\overset{i}{\leadsto}$, (neither player or knowledge update), and because the updates are best reply, it is clear that if $v\overset{i}{\leadsto}v^*$, then $v^*\overset{i}{\leadsto}v^*$. Then, once again, $\overset{i}{\leadsto}=\overset{i}{\leadsto}\overset{i}{\leadsto}$.
      \item If $i=0$ and $j\neq 0$. Let $I=ji$ and $J=ji$. \todo{on
          pourrait le faire formellement, mais ça serait super long !}
        Intuitively, $\overset{jiji}{\leadsto}$ means that player $j$
        changes his strategy, there is a knowledge update, then $j$
        changes again, then there is a knowledge update. It is clear
        that if we remove the first change of $j$, the result will be
        exactly the same,
        i.e. $\overset{jiji}{\leadsto}=\overset{iji}{\leadsto}$
      \item If $j=0$ and $i\neq 0$.  Let $I=jij$ and $J=j$. This is
        symmetrical to the previous point, then we have that
        $\overset{ijij}{\leadsto}=\overset{jij}{\leadsto}$.
      \item If $i\neq 0$ and $j\neq 0$. Then let $I=j$ and
        $J=i$. Since these updates are performed without any knowledge
        update, it is clear that
        $\overset{ij}{\leadsto}=\overset{ji}{\leadsto}$.
      \end{itemize}

    \item We still have to prove that there exists $v$ from which we can reach two different sinks. Let $\tau \in \pstrat$ and $\tilde{\tau}\in\pstrat$ two
      different equilibria of $\graph\game\BPCdyn$ and
      $t=(\tau_{i,j})_{1\le i,j\le n}
      ,\tilde{t}=(\tilde{\tau}_{i,j})_{1\le i,j\le n}\in V_0$ the associated sinks of $\graph\game\leadsto$.
      Let $v=(\sigma_{i,j})_{1\le i,j\le n}\in V$ such that
      $\forall 1\le i \le n,\ \sigma_{i,i}=\tau_i$ and
      $\forall j\neq i,\ \sigma_{i,j}=\tilde{\tau_i}$. In this state,
      all players play the first equilibrium $\tau$ but believe that
      the other players play the second equilibrium
      $\tilde{\tau}$. Clearly, $v\overset{0}{\leadsto} t$ and
      $v\overset{1,\ldots,n}{\leadsto} \tilde{t}$. Then $v$ can reachs two different sinks.
    \end{itemize}

    Then, by Lemma~\ref{Lemma:ab=ba}, $\graph\game\leadsto$ has a non constant \Lfair\ cycle.

  \item We define a partial simulation $\sqsubseteq$ of
    $\graph\game{\leadsto^+}$ by $\graph\game{\BPCdyn^+}$ of domain
    $V_0$ as follows:
    $(\sigma_{i,j})_{1\le i,j\le n}\sqsubseteq (\tau_i)_{1\le i \le
      n}$ if
    $\forall j,\ (\sigma_{i,j})_{1\le i \le n}=(\tau_i)_{1\le i \le
      n}$. The proof that $\sqsubseteq$ is a partial simulation of
    $\graph\game{\leadsto^+}$ by $\graph\game{\BPCdyn^+}$ comes
    directly from the definition of $\BPCdyn$ and~$\leadsto$.\todo{A
      justifier plus? Ca sera long quand même... M.}

  \item Clearly, if $\graph\game\leadsto$ has an $\Lfair$ cycle, it
    means that this cycle contains elements of $V_0$, by definition of
    $\overset{0}{\leadsto}$. Then $\graph\game{\leadsto^+}$ contains a non constant
    cycle with only states of $V_0$. Intuitively, this means that the
    cycle visits states where all the Players are perfectly informed
    about the strategies of all other players, which are the states
    that `corresponds' to those in $\graph\game{\BPCdyn^+}$.

  \item By Proposition~\ref{Cor:weak et co-weak=>same termination}, it
    means that $\graph\game{\BPCdyn^+}$ has a cycle.

  \item As the first cycle built in $\graph\game\leadsto$ was \Lfair,
    it means that, during this cycle, every player has changed his
    strategy, or at least has had the opportunity to change it. It
    means that, in $\graph\game\BPCdyn$, the cycle is fair.\todo{M:pas
      formel du tout mais ça serait si long à faire...}\qedhere
  \end{enumerate}

\end{proof}

\begin{proof}[Proof of Theorem~\ref{Thm:Eqiv condition neighbour}]
  Let $\game$ be a \NOTG. We have to prove that $\PCdyn$ does not fairly terminate for
  $\game$ if and only if $\game$ has an SDW.

In \cite{Griffin02}, Griffin et al. prove a stronger result than
Theorem~\ref{Thm:Griffin}. Indeed, they actually prove that if $\BPCdyn$
  fairly terminates for a \OTG $\game$, then $\game$ has a DW such that: 
\begin{enumerate}
\item $\forall u_i\in U,\ u_i\notin \pi_j\in P,$ if $j\neq i$;
\item $\forall v\notin U,\ v\notin \pi_i\cap h_j,$ for $\pi_i\in P$ and $h_j\in H$;
\item $\forall v\in \pi_i\cap \pi_j$, $\pi_i(v)=\pi_j(v)$.
\end{enumerate} We call such a dispute wheel a DW+.  We show that if a
game has a DW+, this implies the existence of a SDW by showing the
remaining properties:

\begin{enumerate}\setcounter{enumi}{3}
\item $\forall u_i\in U,\ u_i\notin h_j\in H,$ if $j\neq i$ and
  $j\neq i-1$.
\item $\forall v\notin U,\ v\notin h_i\cap h_j,$ for $h_i,\ h_j\in H$
  with $i\neq j$.
\end{enumerate}
Moreover, Theorem~\ref{Thm:DWAB=>not terminates} tells us that having
an SDW implies that $\graph\game\PCdyn$ has a cycle. We conclude the
proof by showing that if $\graph\game\PCdyn$ has a cycle, it has a
fair cycle, because if a player which was not allowed to change during
the non fair cycle decides to change, it will not influence the
choices of the other players, then the cycle remains.

Let us prove this :

\begin{itemize}
\item If $\game$ has a DW+, then we can build a DW such that
  $\forall u_i\in U,\ u_i\notin h_j\in H,$ if $j\neq i$ and
  $j\neq i-1$. On the contrary, let us suppose that $\exists u_i\in U$
  such that $\exists h_j\in H$ with $j\neq i-1;\ j\neq i$ such that
  $u_i\in h_j$. It means that
  $\pi_j=u_jp_j\ldots 0\prec_j u_jr_j\ldots u_i\ldots
  u_{j+1}p_{j+1}\ldots 0=h_j\pi_{j+1}$. By conditions over the
  preferences of the game, it means that, for $h'_j=u_jr_j\ldots p_i$,
  $\pi_j=u_jp_j\ldots 0\prec_j u_jr_j\ldots u_ip_i\ldots
  0=h'_j\pi_i$. Let $\Pi'=(\overrightarrow{U}',P',H')$ where
  $\overrightarrow{U}'=u_1,\ldots, u_j,u_i,\ldots, u_n\ ;\
  P'=\pi_1,\ldots, \pi_j,\pi_i,\ldots,\pi_n\ ;\ H'=h_1,\ldots
  h'_jh_i,\ldots, h_n$. Notice that this DW is still a DW+.
\item If $\game$ has a DW+ satisfying condition 4, then we can build
  a DW such that $\forall v\notin U,\ v\notin h_i\cap h_j,$ for
  $h_i,\ h_j\in H$ with $i\neq j$.  Suppose that
  $\exists v\notin \overrightarrow{U}$ such that $v\in h_i \cup
  h_j$. Let $h_i=u_ir_i\ldots v v_i\ldots u_{i+1}$ and
  $h_j=u_jr_j\ldots v v_j\ldots u_{j+1}$. In particular, as
  $\pi_j\prec_j u_jr_j\ldots v v_j\ldots u_{j+1}p_{j+1}\ldots
  0=h_j\pi_{j+1}$, by conditions over the preferences, it means that,
  for $h'_j=u_jr_j\ldots v v_i\ldots u_{i+1}$, then
  $\pi_j\prec_j u_jr_j\ldots v v_i\ldots u_{i+1}p_{i+1}\ldots
  0=h_j\pi_{j+1}=h'_j\pi_i$. Let
  $\Pi'=(\overrightarrow{U}',P',\overrightarrow{H}')$ where

  $\overrightarrow{U}'=u_1,\ldots, u_j,u_{i+1},\ldots, u_n\ ;\
  P'=\pi_1,\ldots, \pi_j,\pi_{i+1},\ldots,\pi_n\ ;\ H'=h_1,\ldots
  h'_jh_{i+1},\ldots, h_n$. Note that this DW is still a DW+
  satisfying condition 4.

  By iterating this process, we can build a \DWAB.

\item First of all, let $\sigma_1\to\ldots\to \sigma_n$, a (not fair)
  cycle in $G$. Let us suppose that a player $i$ cannot change his
  strategy during this cycle.  Let $\sigma_{1,i}$ the best reply of
  player $i$ to $\sigma_1$. By hypothesis of Neighbour game,
  $\tilde{\sigma_1}\to \tilde{\sigma_n}$ is a fair cycle, with
  $\tilde{\sigma{j,i}}=\sigma_{1,i}$ and
  $\forall k\neq i,\ \tilde{\sigma{j,k}}=\sigma{j,k}$\qedhere
\end{itemize}
\end{proof}

\begin{proof}[Proof of Proposition~\ref{prop:disagree}]
  By Theorem~\ref{Thm:Eqiv condition neighbour}, if $\game^{\mathrm{DIS}}$ is a
  minor of $\game$, then $\PCdyn$ does not fairly terminates for
  $\game$. Moreover, still by Theorem~\ref{Thm:Eqiv condition
    neighbour}, if $\PCdyn$ does not fairly terminate, then it means
  that $\game$ has an SDW. It remains to prove that if $\game$ has an
  SDW, $\game$ has $\game^{\mathrm{DIS}}$ as a minor.

  Let $D=(U,P,H)$ be the SDW of $\game$ with
  $U=(u_1,\ldots,u_k),\ P=(\pi_1,\ldots,\pi_k)$ and
  $H=(h_1,\ldots,h_k)$. If $k>2,$ let
  $D'=((u_1,u_2),(\pi_1,\pi_2),(h_1,h_2'))$ with
  $h_2'=h_2h_3\ldots h_k$. We do have that $\pi_1\prec_1 h_1\pi_2$ and
  $\pi_2\prec_2h_2'\pi_1\eq_2 h_2\pi_3$ because $\game$ is an
  \NOTG. In particular, $\game^{\mathrm{DIS}}$ is a minor of $\game$.
%
%
%
%
%
%
\end{proof}

\end{document}